%% file: main.tex
\documentclass[11pt,letterpaper]{article}
\usepackage[margin=1in]{geometry}

\usepackage{CJK}

\usepackage{amsthm}
\usepackage{amsmath,amsfonts,amssymb,bm}
\usepackage{graphicx} 
\usepackage{xcolor}

\usepackage{float} 
\usepackage{subfigure} 
\usepackage{framed}
\usepackage{array}

\input{math-command}

\newtheorem{theorem}{Theorem}[section]
\newtheorem{lemma}{Lemma}[section]

\newtheorem{claim}{Claim}[section]

\newtheorem*{question*}{Question}

\newtheorem{definition}{Definition}

\definecolor{WildStrawberry}{RGB}{255,67,164}
\definecolor{Dousha}{RGB}{47, 136, 67}
\definecolor{BlueBerry}{RGB}{44, 96, 189}

\title{On Robustness to $k$-wise Independence of Optimal Bayesian Mechanisms}

\author{Nick Gravin \and Zhiqi Wang}
\date{}

\input{notations}

\begin{document}

\maketitle
\begin{abstract}
\input{sec/abstract}    
\end{abstract}

\section{Introduction}
\label{sec: intro}
\input{sec/intro}
\subsection{Model: $k$-wise Independence and Robustness} 
\label{sec: model}
\input{sec/setting}
\subsection{Our results}
\label{sec: results}
\input{sec/results}
\subsection{Related Work}
\label{sec: related}
\input{sec/related2}
\section{Preliminaries}
\label{sec: prelim}
\input{sec/model}
\section{Independence-Robustness Degree of Myerson's Mechanism}
Following the textbook~\cite{Krishna02} by Krishna, we distinguish between \emph{Myerson's mechanism} for nonidentical agents and \emph{Myerson's auction} for i.i.d. bidders. According to Krishna, the former cannot be called an auction as it does not satisfy the anonymity requirement and should instead be placed in a broader class of mechanism-like institutions. The latter is equivalent to the broadly used auction format of the second price auction with anonymous reserve $\AR$. This section shows that the robustness degree of independence for Myerson's mechanism is exactly $3$. 
We first show that Myerson's mechanism is not pairwise-robust even for regular marginal distributions in the next Section~\ref{sec: nega}. Then we prove in Section~\ref{sec: 3wi} that Myerson's mechanism is $3$-wise-robust for any (even irregular) marginals $(\disti)_{i\in[n]}$. 
\subsection{Myerson's Mechanism is not Pairwise-Robust}
\label{sec: nega}
\input{sec/nega}
\subsection{$3$-Wise-Robustness}
\label{sec: 3wi}
\input{sec/3wi}

\section{Pairwise-Robustness of Myerson's Auction}
\label{sec: bounds}
\input{sec/bounds}

\subsection{Identical Marginal Distributions}
\label{sec: iden}
\input{sec/iden}

\subsection{Nonidentical Marginal Distributions}
\label{sec: ar}
\input{sec/ar}

\bibliographystyle{plain}
\bibliography{main}
\appendix

\section{Proof of Claim~\ref{clm: T}}
\label{app: corollary}
\input{sec/appendixnew}

\section{Proof of Claim~\ref{clm: max>p}}
\label{app: case3}
\input{sec/appendix3}
\end{document}

%% file: math-command.tex
\newcommand{\etal}{\emph{et al.} }

\def\1{\bm{1}}

\makeatletter
\newcommand{\ve}{\@ifnextchar\bgroup{\velong}{{\bm{e}}}}
\newcommand{\velong}[1]{{\bm{#1}}}
\makeatother

\def\vp{{\bm{p}}}
\def\vq{{\bm{q}}}

\DeclareMathAlphabet{\mathsfit}{\encodingdefault}{\sfdefault}{m}{sl}
\SetMathAlphabet{\mathsfit}{bold}{\encodingdefault}{\sfdefault}{bx}{n}

\def\gM{{\mathcal{M}}}

\newcommand{\eat}[1]{}

\def\d{\mathrm{d}}

%% file: notations.tex
\newcommand{\be}{\begin{equation}}
\newcommand{\ee}{\end{equation}}
\newcommand{\beq}{\begin{equation*}}
\newcommand{\eeq}{\end{equation*}}

\newcommand{\R}{\mathbb{R}}

\newcommand{\Rev}{\textup{\textsc{Rev}}}

\newcommand{\eps}{\varepsilon}
\newcommand{\ind}[1]{\mathbb{I}\hspace{-0.1em}\left[\vphantom{\sum}#1\right]}

\newcommand{\AutoAdjust}[3]{\mathchoice{ \left #1 #2  \right #3}{#1 #2 #3}{#1 #2 #3}{#1 #2 #3} }
\newcommand{\Xcomment}[1]{{}}

\newcommand{\InParentheses}[1]{\AutoAdjust{(}{#1}{)}}
\newcommand{\InBrackets}[1]{\AutoAdjust{[}{#1}{]}}
\newcommand{\Ex}[2][]{\operatorname{\mathbf E}_{#1}\InBrackets{#2}}
\newcommand{\Exlong}[2][]{\operatornamewithlimits{\mathbf E}\limits_{#1}\InBrackets{#2}}

\newcommand{\Prx}[2][]{\operatorname{\mathbf{Pr}}_{#1}\InBrackets{#2}}

\newcommand{\eqdef}{\overset{\mathrm{def}}{=\mathrel{\mkern-3mu}=}}
\newcommand{\vect}[1]{\ensuremath{\mathbf{#1}}}

\newcommand\restr[2]{{
  \left.\kern-\nulldelimiterspace 
  #1 
  \vphantom{\big|} 
  \right|_{#2} 
  }}
\def\prob{\Prx}

\newcommand{\AR}[1][]{\textsf{AR}_{#1}} 

\newcommand{\opt}{\textsf{OPT}}

\newcommand{\dist}{\mathcal{F}}

\newcommand{\dists}{\bm{F}}
\newcommand{\disti}[1][i]{{\dist_{#1}}}

\newcommand{\pdf}{f}
\newcommand{\pdfi}[1][i]{{\pdf_{#1}}}

\newcommand{\uniform}{\texttt{Uni}}

\newcommand{\vecX}{\vec{X}}

\newcommand{\supp}{\mathop{\textsf{supp}}}

\newcommand{\mech}{\mathcal{M}}

\newcommand{\val}{v}
\newcommand{\vals}{\vect{\val}}

\newcommand{\vali}[1][i]{{\val_{#1}}}

\newcommand{\util}{u}

\newcommand{\utili}[1][i]{\util_{#1}}

\newcommand{\pay}{p}

\newcommand{\payi}[1][i]{{\pay_{#1}}}

\newcommand{\alloc}{x}

\newcommand{\alloci}[1][i]{{\alloc_{#1}}}

\newcommand{\bid}{b}
\newcommand{\bids}{\vect{\bid}}

\newcommand{\bidi}[1][i]{{\bid_{#1}}}

\newcommand{\virt}{\varphi}
\newcommand{\virti}[1][i]{\virt_{#1}}

\newcommand{\virtir}{\overline{\virt}}
\newcommand{\virtiri}[1][i]{\virtir_{#1}}

\newcommand{\myerson}{\textsf{Myer}}

\renewcommand{\d}{\mathrm{d}}

\renewcommand{\vec}{\vect}

\newcommand{\Lowerboundone}{\textsc{LB}_1}
\newcommand{\Lowerboundtwo}{\textsc{LB}_2}
\newcommand{\Lowerboundthree}{\textsc{QR}_{_{\textsc{LB}}}}

\newcommand{\ER}{ER}

\newcommand{\Qoneind}{Q_1^{ind}}

\newcommand{\Qtwoind}{Q_2^{ind}}

\newcommand{\Qone}{Q_1}
\newcommand{\hQone}{\widehat{Q}_1}
\newcommand{\Qtwo}{Q_2}

\newcommand{\Quan}{q}
\newcommand{\Quani}[1][i]{\Quan_{#1}}
\newcommand{\ubqi}[1][i]{\overline{\Quan}_{#1}}
\newcommand{\vubq}{\overline{\vect{\Quan}}}

\newcommand{\distind}{\dists_{ind}}

\newcommand{\distfam}{\mathfrak{F}}
\newcommand{\distowi}{\mathfrak{F}_{1w}}
\newcommand{\distpwi}{\mathfrak{F}_{2w}}
\newcommand{\disttwi}{\mathfrak{F}_{3w}}
\newcommand{\distkwi}[1][k]{\mathfrak{F}_{#1\text{-}w}}

\newcommand{\reserve}{r}
\newcommand{\ww}{{r_{ex}}}
\newcommand{\ttau}{{\tau_{ex}}}
\newcommand{\valo}{\val_{0}}
\newcommand{\mr}{\reserve^{*}}
\newcommand{\mri}[1][i]{\reserve^{*}_{#1}}
\newcommand{\kk}{\gamma}
\newcommand{\pp}{\overline{p}}

\newcommand{\Payment}{\mathcal{P}}
\newcommand{\probi}[1][i]{p_{#1}}

\newcommand{\momone}{s}
\newcommand{\moww}{s_{0}}
\newcommand{\hmoww}{\widehat{s}}

\newcommand{\sumtau}{s(\tau)}

\newcommand{\tailrev}{\texttt{Tail}}
\newcommand{\corerev}{\texttt{Core}}

\newcommand{\event}{\mathcal{E}}
\newcommand{\exvi}[1][i]{w_{#1}}
\newcommand{\qi}[1][i]{q_{#1}}
\newcommand{\revi}[1][i]{\textsc{Rev}_{#1}}

\newcommand{\probone}{p_1}
\newcommand{\probtwo}{p_2}
\newcommand{\alfa}{\alpha}
\newcommand{\ttt}{t}
\newcommand{\mm}{m}

%% file: sec/abstract.tex
This paper reexamines the classic problem of revenue maximization in single-item auctions with $n$ buyers under the lens of the robust optimization framework. The celebrated Myerson's mechanism is the format that maximizes the seller's revenue under the prior distribution, which is mutually independent across all $n$ buyers. 
As argued in a recent line of work (Caragiannis \etal 22), (Dughmi \etal 24), mutual independence is a strong assumption that is extremely hard to verify statistically, thus it is important to relax the assumption. 

While optimal under mutual independent prior, we find that Myerson's mechanism may lose almost all of its revenue when the independence assumption is relaxed to pairwise independence, i.e., Myerson's mechanism is not pairwise-robust.
The mechanism regains robustness when the prior is assumed to be 3-wise independent. 
In contrast, we show that second-price auctions with anonymous reserve, including optimal auctions under i.i.d. priors, lose at most a constant fraction of their revenues on any regular pairwise independent prior.
Our findings draw a comprehensive picture of robustness to $k$-wise independence in single-item auction settings.

%% file: sec/intro.tex

The notion of approximation is central in the algorithmic research on Bayesian mechanism design. Since the seminal paper by Hartline and Roughgarden~\cite{HartlineR09}, it is common (e.g., in~\cite{ChawlaHMS10, RoughgardenTQ12, AlaeiFHH13, FeldmanGL15, BrustleCWZ17, FengHL19, jin2020tight, FeldmanGGS22}) to study how well a certain mechanism approximates optimum for a given prior distribution. Perhaps one of the most compelling motivations (see the philosophy of approximation in~\cite{HartlineMD}) for employing a simple solution in place of a highly customized optimum is \emph{robustness}. Many approximately optimal mechanisms are known to be robust to out-of-model considerations, such as risk aversion, after-market effects, timing, collusion, etc. E.g., sequential posted pricing is robust to potential collusion among bidders, while second price auction is not. 
Even more basic robustness considerations concern the statistical modeling of the Bayesian prior. The existing literature mostly
assumes that the prior distribution is independent across agents and/or items and does not investigate if the results are robust to this mutual independence assumption.

On the other hand, the area of robust mechanism design offers a counterpoint to the Bayesian mechanisms, which have perfect knowledge about the prior distribution $\dists$. Specifically, robust mechanisms have ambiguity set $\distfam$ for the possible prior distribution and their goal is to identify the optimal solution in the worst-case over all priors $\dists\in\distfam$. One such example is the correlation-robust framework~\cite{carroll2017robustness, GravinL18}, in which only marginal distributions $(F_i)_{i\in[m]}$ are known, but no other assumption is made about the joint prior $\dists$. 
The ambiguity sets considered in the literature are typically quite large\footnote{E.g., for $n=10$ marginal distributions with support size $2$, there are only $20$ linear constraints in the vector space of dimension $2^{10}$ to describe the ambiguity set in the correlation-robust framework.}. These modeling choices lead to extreme worst-case prior distribution $\dist^*\in\distfam$, such as perfect positively correlated prior $\dist^*$ in~\cite{BeiGLT19}, and generally to rather pessimistic results. This makes most of the existing robust mechanism design models difficult to use in scenarios with no or only weak correlations and thus not applicable to most of the existing Bayesian settings.

A notable exception is a recent work by Caragiannis et al.~\cite{caragiannis2022relaxing}. They observed that the mutual independence can be relaxed to much weaker pairwise-independence assumption in some existing Bayesian mechanism design results, such as Bayesian monopolist problem~\cite{BabaioffILW20}, $O(1)$-prophet inequality, and $O(1)$-approximate sequential posted pricing.
This observation inspired us to ask: 

\begin{question*}
To what degree can the mutual independence assumption be relaxed in a single-item auction setting, and how robust are optimal and approximately optimal mechanisms to this assumption?
\end{question*}
To illustrate the importance of relaxing the mutual independence assumption with a concrete example, imagine that a platform observes a large amount of data from past ad auctions. 
There could be tens to thousands of different bidders with only a small fraction participating in each auction. Let us assume that the platform has enough data to learn the value distribution of each individual bidder with high confidence. The next common step is to look at the covariance matrix.
For the sake of example, let us further assume that the covariance matrix perfectly fits our hypothesis that bidders have independent values. Even in this hypothetical situation the platform can only confirm that the joint prior distribution is \emph{pairwise independent} rather than \emph{mutually independent}.

%% file: sec/setting.tex
We wish to understand how robust are various Bayesian mechanisms to precise modeling of the prior distributions and specifically to the strong probabilistic assumption of \emph{mutual independence}. We consider the classic and well-studied setting of the single-item auction with $n$ bidders whose values are distributed according to known marginal distributions $(\disti)_{i\in[n]}$.
While the joint prior $\dists$ may appear to be mutually independent distribution $\distind=\prod_{i\in[n]}\disti$ from a limited number of observations, it is impossible to confirm by doing pairwise statistical analysis such as computing correlation coefficients. We assume instead that the actual prior distribution $\dists$ belongs to the family of pairwise independent distributions $\dists\in\distpwi$ with known single-dimensional marginals $\disti$ for each $i\in[n]$. However, we do not know which $\dists\in\distpwi$ it is. The mutually independent distribution $\distind\in\distpwi$ plays a special role in $\distpwi$ as an obvious comparison point for any $\dists\in\distpwi$. Ideally, we would like to have a truthful mechanism $\mech$ to 
perform well on any $\dists\in\distpwi$. That is, the mechanism's expected revenue $\Ex[\vals\sim\dists]{\Rev(\vals)}$ should be within a constant factor from $\Ex[\vals\sim\distind]{\Rev(\vals)}$ for any number of agents $n$, any set of marginals $(\disti)_{i\in[n]}$ and any pairwise independent prior $\dists\in\distpwi$. 
Formally, 
\begin{definition}
\label{def: pwi}
A mechanism $\gM$ is \emph{pairwise-robust}, if there exists a constant $c$ such that for all pairwise independent joint distribution $\dists\in\distpwi$ the expected revenues 
\begin{equation*}
    c\cdot \gM\left(\dists\right) \ge 
    \gM\left(\distind\right ).
\end{equation*}
\end{definition}
One can naturally extend Definition~\ref{def: pwi} to a more refined family of $3$-wise independent prior distributions $\disttwi$ and even consider the whole hierarchy of the families of $k$-wise independent priors for $2\le k\le n$: $\distpwi\supset\disttwi\supset\ldots\supset\distkwi\supset\ldots\supset\distkwi[n]=\{\distind\}$. 
\begin{definition}[$k$-wise-robustness]
\label{def: 3wi}
A mechanism $\gM$ is \emph{$k$-wise-robust}, if there exists a constant $c(k)$ such that for all $k$-wise independent joint distribution $\dists\in\distkwi$ the expected revenues 
\begin{equation*}
    c(k)\cdot \gM\left(\dists\right) \ge 
    \gM\left(\distind\right ).
\end{equation*}
\end{definition}
We say that a mechanism $\gM$ is \emph{over-optimizing} for the mutually independent prior, if $\gM$ is not pairwise-robust. If $\gM$ is not $(k-1)$-wise-robust for $k\ge 3$, then we say that $\gM$ is \emph{over-optimizing up to degree $k$}. If in addition (to being not $(k-1)$-wise-robust) $\gM$ is $k$-wise-robust, we say that the mechanism's \emph{degree of independence-robustness} is $k$. 

We include in the hierarchy of $k$-wise independent families the family $\distowi$ of distributions with given marginals $(\disti)_{i\in[n]}$, but with no restrictions on the correlation structure of the joint prior $\dists\in\distowi$:
$\distowi\supset\distpwi\supset\ldots\supset\distind$. The family $\distowi$ is central in the correlation-robust framework~\cite{carroll2017robustness, GravinL18} that aims to identify optimal mechanisms for the worst-case prior $\dists$ over the ambiguity set $\dists\in\distowi$, or approximately optimal mechanisms~\cite{BeiGLT19, BabaioffFGLT20}. 
One can also consider $1$-wise-robustness of a mechanism $\gM$ instead of approximate correlation-robustness (in the former case the benchmark is $\gM(\distind)$ and in the latter case the benchmark is $\max_{\gM'}\min_{\dists\in\distowi}\gM'(\dists)$).  

\paragraph{Robustness to $k$-wise independence in other contexts.}
Markov, Chebyshev, and Chernoff upper bounds on the tail probability of the sum 
of $n$ bounded random variables $\Prx[\vecX\sim\dists]{\sum_{i\in[n]}X_i-
\Ex{\sum_{i\in[n]}X_i}\le\delta}$ are fundamental tools in probability theory.
The Markov inequality only needs information about the marginal distributions of 
the random variables $\vecX=(X_i)_{i\in[n]}$ and, hence, can be viewed as $1$-wise robust and hold for priors $\dists\in\distowi$. The Chebyshev inequality, in addition to the information about marginal distributions, also requires that $\vecX\sim\dists$ is pairwise independent $\dists\in\distpwi$, i.e., the Chebyshev bound is pairwise-robust and holds for priors $\dists\in\distpwi$. The Chernoff bound requires mutual independence of $\dists$ and fails to hold for pairwise-independent priors $\dists\in\distpwi$. Thus, in our language, it over-optimizes for the mutually independent distribution $\distind$. On the other hand, there are versions of the Chernoff bound, e.g.,~\cite{SchmidtSS93}, that relax the independence assumption to $k$-wise independence where the tail bound rapidly decreases with $k$.

The $k$-wise independence plays an important role in pseudorandom generators as a much easier property to achieve than mutual independence. Furthermore, 
$k$-wise independence for a constant $k$ can be tested with only polynomially many samples~\cite{AlonAKMRX07, RubinfeldX10}. Specifically, the testing algorithm can distinguishing the case $\dists\in\distkwi$ from the 
case when the total variation distance $d_{\texttt{TV}}\left(\dists,\distkwi\right)\ge\eps$ is at least $\eps$ with polynomial in the maximum support size $\max_i|\supp(\disti)|$ and $\frac{1}{\eps}$ number of samples. Although, the polynomial degree grows linearly with $k$.

In the mechanism design context, a recent work~\cite{caragiannis2022relaxing} shows that in certain Bayesian mechanism design settings such as monopolist problem~\cite{BabaioffILW20}, sequential posted pricing and single item prophet inequality the mutual independence assumption can be relaxed to pairwise independence. I.e., those mechanisms are pairwise-robust. The follow-up work~\cite{dughmi2023limitations} shows pairwise non-robustness of the multi-choice matroid prophet inequality and contention resolution schemes. These negative results are caused by the complex combinatorial structure of general matroid set systems. 
On the other hand, the sequential posted pricing is known to be robust to many different out-of-model phenomena. E.g., it is robust to collusion among bidders, timing effects, and according to~\cite{BeiGLT19} is correlation-robust (not to be confused with $1$-wise robustness).

%% file: sec/results.tex
\begin{enumerate}
    \item Given the successful stories of relaxing mutual to pairwise independence in~\cite{BabaioffILW20} for the monopolist problem and in~\cite{caragiannis2022relaxing} for sequential posted pricing in single-item auction, it is natural to conjecture that Myerson's mechanism is as well pairwise-robust. We were rather surprised to find a counter-example to this conjecture. Specifically, we construct an instance of the pairwise independent joint prior $\dists\in\distpwi$ with nonidentical regular marginal distributions $(\disti)_{i\in[n]}$ such that $\Ex[\vals\sim\distind]{\myerson(\vals)}=\Omega(n)\cdot\Ex[\vals\sim\dists]{\myerson(\vals)}.$ 
    \item We show that Myerson's mechanism is $3$-wise robust for any (even irregular) set of marginals $(\disti)_{i\in[n]}$. Hence, we establish that the independence-robustness degree of Myerson's mechanism is exactly $3$.
    \item We consider Myerson's auction, i.e., a Myerson's mechanism in the symmetric environment with identical marginals $\disti=\dist$ for each $i\in[n]$, which is equivalent to second-price auction with monopoly reserve $\AR(\mr)$. We prove that Myerson's auction for regular marginal distribution $\dist$ is pairwise-robust with a constant $2.63$. I.e., we establish a separation in independence-robustness degree between symmetric and nonsymmetric environments. 
    \item Furthermore, we consider the family of commonly used second-price auctions with anonymous reserve $\AR(\reserve)$ in the nonsymmetric environment. We prove that $\AR(\reserve)$ with any reserve price $\reserve$ is pairwise-robust for any set of regular marginals $(\disti)_{i\in[n]}$. This result highlights the advantage of auctions\footnote{The textbook by Krishna~\cite{Krishna02} differentiates between auctions-like institutions (must be anonymous) and general mechanism-like institutions.}, which are pairwise-robust, over optimal mechanisms, which are only $3$-wise robust, in the single-item environment. 
\end{enumerate}
Our results are summarized in the following table. AR is short for the second price with anonymous reserve, SPP is short for the sequential posted price, and AP is short for the anonymous price.

\begin{table}[h]
\centering
\begin{tabular}{|>{\centering\arraybackslash}m{3cm}|>{\centering\arraybackslash}m{2cm}|>{\centering\arraybackslash}m{2cm}|>{\centering\arraybackslash}m{2cm}|>{\centering\arraybackslash}m{2cm}|>{\centering\arraybackslash}m{2cm}|}
\hline
\textbf{Mechanism} & Myerson &  Myerson $(\disti = \dist)$ & AR & SPP & AP\\ \hline
\textbf{Degree of Independence - Robustness} &  3,~$*$ & 2,~$*$& 2,~$*$ & 2,~\cite{caragiannis2022relaxing}& 2,~\cite{caragiannis2022relaxing} \\ \hline
\end{tabular}
\caption{Degree of independence-robustness for various mechanisms. $\disti = \dist$ denotes that all marginals are identical. The sign ``$*$" denotes that the result is obtained in this work.}
\end{table}

%
%
%
%
%
In his book~\cite{HartlineMD}, Hartline comments on the philosophy of approximation in mechanism design:``... optimal mechanisms may be overly influenced by insignificant-seeming modeling choices. ... a market analysis is certainly going to be noisy and then exactly optimizing a mechanism to it may “overfit” to this noise''. We focus on the basic setting of revenue maximization for single-item auctions and demonstrate that the optimal mechanism may be highly over-fitting even with perfectly inferred marginal distributions to the statistically unverifiable assumption of mutual independence, while ``simple'' approximate solutions, such as $\AR(\reserve)$ and sequential posted pricing, are pairwise-robust. I.e., our notion of $k$-wise robustness brings a new quantifiable measure of how optimal mechanisms can be inferior to robust approximate solutions without adding any out-of-model elements such as, e.g., timing effects or possible collusion among bidders.
%
%

\paragraph{Techniques.} The important part of our results is an example of pairwise-independent distribution $\dists$ with $\Omega(n)$ separation between $\myerson(\dists)$ and $\myerson(\distind)$. While verifying its correctness is rather straightforward, it was not easy to construct. We are not aware of any work in mechanism design that relies specifically on $3$-wise independence but fails to work for pairwise independent priors. Thus, we had to develop new approach to proving $3$-wise robustness of the Myerson's mechanism.

Our positive results on pairwise-robustness for $\AR(\reserve)$ use similar integral representations to Caragiannis et al.~\cite{caragiannis2022relaxing}. However, they only consider simple events $\event_1(\tau)$ that at least one bidder exceeds a threshold $\tau$ in the sequential posted pricing, while $\AR(\reserve)$ has events $\event_2(\tau)$ that at least two bidders are above the threshold $\tau$. Handling $\event_2$ poses a significantly bigger challenge than $\event_1$. Indeed, Caragiannis et al. could prove a constant factor gap between probabilities in mutually and pairwise independent cases for any event $\event_1(\tau)$, which is no longer true for $\event_2(\tau)$. Thus, we have to carefully compare different parts of respective integral representations for mutually independent and pairwise independent priors. 

%% file: sec/related2.tex
The body of work on revenue maximization in single-item auctions is too large to describe within the scope of our paper. We mention a few papers in the Bayesian framework that are most closely related to our setting.
On the computational side, Papadimitriou and Pierrakos~\cite{papadimitriou2011optimal} proved NP-hardness of finding revenue optimal mechanism under ex-post individual rationality constraint, while Dobzinski et al.~\cite{dobzinski2011optimal} showed that optimal randomized mechanism could be found polynomial in the prior's support size, and a simple and approximately optimal lookahead auction by Ronen and Saberi~\cite{RonenS02} is easy to compute. When the prior $\distind$ is independent, a few papers give multiple approximation guarantees on the revenue of simple auction formats such as sequential posted-pricing, second-price auction with an anonymous reserve, anonymous pricing, and second-price auction with monopoly reserves to the optimal revenue of Myerson's auction~\cite{CorreaSZ19,DuttingFK16,HartlineR09,jin2020tight}.
Finally, there is a significant body of work on optimization from samples initiated in~\cite{ColeR14} that aims to relax dependency of the mechanism on the precisely given Bayesian prior. They still assume that the the prior is a product distributions $\dists=\prod_{i\in[n]}\disti$,
but the marginals $(\disti)_{i\in[n]}$ are not explicitly given, but can be accessed via a limited number of independent samples.


Our model falls under the umbrella of robust mechanism design, which has received significant attention in recent years (see e.g.,~\cite{BabaioffFGLT20,BandiB14,BeiGLT19,bergemann2011robust,caragiannis2022relaxing,carroll2015robustness,carroll2016robust,DuttingRT19,garrett2014robustness,GravinL18}).
This line of work generally relies on a robust optimization framework, tracing back to the classic work in Operation Research by Scarf~\cite{Scarf1958minmax}. 
I.e., the objective is to optimize the mechanism's expected performance in the worst case across all prior distributions within a given ambiguity set $\distfam$ of possible priors $\dists\in\distfam$.
It is typical in this work to consider large ambiguity sets $\distfam$ 
with minimal information about marginal distributions described, e.g., by 
the distribution's support or a constant number of moments. The most 
similar to our setting is the correlation-robust 
model~\cite{carroll2017robustness, GravinL18} that assumes complete knowledge of the marginals, but no restrictions on the correlation 
structure. The papers~\cite{carroll2017robustness, GravinL18, BabaioffFGLT20} study a different setting of multi-item monopoly problem, or in case of~\cite{DuttingRT19} contract design, while \cite{BeiGLT19,HeL22} look at similar to us revenue maximization single-
item auction. The paper~\cite{BeiGLT19} shows that sequential posted pricing is approximately optimal in the correlation-robust framework (not to be confused with $1$-wise robust), while~\cite{HeL22} examine the symmetric case where the number of bidders $n$ goes to infinity.

%% file: sec/model.tex
We consider the following standard Bayesian single-item auction. The auctioneer sells one item to $n$ bidders, where each bidder $i\in[n]$ has a private value $\vali$ for the item drawn from a known marginal distribution $\vali\sim\disti$. 
We also use $\disti(\tau)=\Prx{\vali\le \tau}$ to denote cumulative density function of $\disti$ and define quantile function
$\Quani(\tau)\eqdef \Prx{\vali\ge \tau}= 1-\disti(\tau)$.
We denote the probability density function of $\disti$ as
$\pdfi(x)=\frac{\mathrm{d}}{\mathrm{d}x}\disti(x)$. 
The seller's goal is to maximize her expected revenue over truthful auctions. By the revelation principle, we assume that the auctioneer solicits bids  
$\bids=(\bidi)_{i\in[n]}$ and then decides on the allocation 
$(\alloci(\bids))_{i\in[n]}$ and payments 
$(\payi(\bids))_{i\in[n]}$. The auction is truthful, if each bidder $i\in[n]$ maximizes her utility 
$\utili(\vali,\bidi)=\vali\cdot\alloci(\bidi)-\payi(\bidi)$ by 
bidding truthfully $\bidi=\vali$ and receives non-negative utility 
$\utili(\vali,\vali)\ge 0$. The auctioneer's revenue is $\Ex[\vals]{\sum_{i\in[n]}\payi(\vals)}$.

The seminal Myerson's mechanism~\cite{myerson1981optimal} maximizes the auction's expected revenue $\myerson(\distind)\eqdef
\Ex[\vals\sim\distind]{\sum_{i\in[n]}\payi(\vals)}$ in the case when bidders' values are \emph{mutually independent} where $\distind=\prod_{i\in[n]}\disti$. 
The mechanism allocates the item to the bidder with the highest non-negative \emph{virtual value} $\virti(\vali)$ and the winner has to pay the threshold bid that would still win the auction. The virtual value is defined as
$\virti(\vali)=\vali-\frac{1-\disti(\vali)}{\pdfi(\vali)}$ for \emph{regular} distributions, i.e., $\virti(\vali)$ is monotone increasing in $\vali$. 
For irregular distribution, an ironing procedure is applied and we have the ironed virtual value $\virtiri(\vali)$ so that $\virtiri(\vali)$ is monotone increasing in $\vali$.
In the I.I.D. case, when agents' values are drawn from the same regular distribution $\disti=\dist$, Myerson's mechanism is the same as the second-price auction with anonymous reserve $\reserve\in\R_{\ge 0}$ (denoted as $\AR(\reserve)$), where usually\footnote{Certain prior distributions $\dist$ may have support $[\underline{\val},\overline{\val}]$ with $\virt(\underline{\val})>0$, in which case the reserve $\reserve=\underline{\val}$.} $\virt(\reserve)=0$. 
The $\AR(\reserve)$ auction allocates the item to the highest bidder $i$ with the bid above the reserve $\bidi\ge \reserve$ and charges her either the second highest bid or the 
reserve price, whichever is larger. We use $\AR^{ind}(\reserve)$ to denote the auction's expected revenue on the mutually independent prior $\distind$.
It is also known that
$\AR^{ind}(\reserve)$ is a constant approximation to $\myerson$ for mutually independent prior $\distind$
for any nonidentical marginal regular distributions $(\disti)_{i\in[n]}$ \cite{jin2020tight}. 

\paragraph{Regularity \& Ironing.}
For the most part of this paper we will require that all marginal distributions  
$(\disti)_{i\in[n]}$ are regular. This assumption makes sense, as the revenue gap between Myerson's mechanism and $\AR(r)$ may be arbitrarily large for irregular distributions $(\disti)_{i\in[n]}$. 

Another complication with Myerson's mechanism and ironed distributions is that the mechanism in this case is often not unique. E.g., consider the case of irregular identical marginals $\disti=\dist$ such that $\dist$ has to be ironed on an interval $[a,b]$ above the monopoly reserve $r=\virt^{-1}(0)$. Then all bidders with values $\vali\in[a,b]$ have the same ironed virtual value and the mechanism might need to break ties among several of them with non-zero probability. 
It can break ties in many ways, e.g., uniformly at random or lexicographically. 
In fact, the optimal mechanism may even use
bidders' values to break ties. E.g., when the marginals are the same equal revenue distribution $\dist$ supported on $[1,10]$, the $\AR(\mr)$ with monopoly reserve $\mr=\virti^{-1}(0)=1$ is an optimal mechanism (here all values $\vali\in[1,10)$ give the same virtual value $\virt(\vali)=0$, but the tie-breaking is such that the bidder $i$ with the highest value $\vali$ wins the auction).  
When the prior distribution $\distind$ is mutually independent, all versions of Myerson's mechanism give the same expected revenue. However, this is no longer the case for every joint pairwise-independent prior $\dists\in\distpwi$.
The same issue arises for extremal regular marginals, such as equal revenue distributions and/or marginal distributions with atoms. Thus, when discussing 
$k$-wise-robustness of Myerson's mechanism, we shall allow the auctioneer to pick a specific most convenient tie-breaking rule in Myerson's mechanism. On the other hand, the auction format is completely fixed for the second-price auction $\AR(\reserve)$ with anonymous reserve, i.e., tie-breaking does not matter in this case.

%% file: sec/nega.tex
We show here that Myerson's mechanism is over-optimizing for the independent prior $\distind$.
\begin{theorem}
Myerson's mechanism is not pairwise-robust for regular marginals $(\disti)_{i\in[n]}$.
\end{theorem}
\begin{proof} 
We construct an example with $n+1$ bidders as follows. 
Values of each bidder $i\in[n+1]$ are distributed according to slightly perturbed \emph{Equal Revenue} distributions $\ER[\underline{\val},\overline{\val}]$ supported on $[\underline{\val},\overline{\val}]$. 
Recall that cumulative density function $\dist(x)$ of $\ER[\underline{\val},\overline{\val}]$ 
satisfies $s\cdot \Prx[\val]{\val\ge s}=s\cdot(1-\dist(s))=\underline{\val}\cdot 1$ for any $s\in[\underline{\val},\overline{\val}]$;
it means that $\ER$ has a point mass of weight $\underline{\val}/\overline{\val}$ at $\val = \overline{\val}$ and continuous distribution with the probability density function $f(s)=\underline{v}/s^2$ for $s\in[\underline{\val},\overline{\val})$. It is a regular distribution with virtual value $\virt(\val)=\val-\frac{\underline{\val}/\val}{\underline{\val}/\val^2}=0$ for
$\val\in[\underline{\val},\overline{\val})$ and $\virt(\overline{\val})=\overline{\val}$.

In our example the first $n$ bidders have the same marginal distribution $\ER[\frac{1}{n},1]$. The last bidder's marginal distribution is a perturbed $\ER[n,n^2]$. Namely, $\vali[n+1]=\eps+
\val$, where $\val\sim\ER[n,n^2]$ and $\eps>0$ is a very small number. The virtual value of the last bidder is $\virti[n+1](\vali[n+1])=\eps>0$ for $\vali[n+1]\in[n+\eps,n^2+\eps)$, and 
$\virti[n+1](n^2+\eps)=n^2+\eps$, i.e., $\disti[n+1]$ is also a regular distribution.

Myerson's mechanism allocates to one of the first $n$ bidders if 
and only if at least one of them has value $\vali=1$ and the last bidder's value is strictly smaller than $\vali[n+1]< n^2+\eps$. Otherwise, the last bidder wins the auction. The seller gets a payment smaller than or equal to $1$ in the former case; and in the latter case, she either collects the minimal price $\payi[n+1]=n+\eps$ from the $(n+1)$-th bidder when the first $n$ bidders have values all strictly smaller than $1$ or gets $\payi[n+1]=n^2+\eps$ when at least one of the first $n$ bidders has value $\vali=1$. The Myerson auction is optimal for the mutually independent prior distribution $\distind$. Hence, it extracts at least as much revenue as always selling the item to the last bidder at $n+\eps$ price. I.e., $\myerson(\distind)\ge n$.

Next, we describe a pairwise-independent distribution $\dists\in\distpwi$ such that the Myerson auction achieves at most a constant revenue $O(1)$. The distribution $\dists$ is as follows:
\begin{description}
\item[with probability $\frac{1}{n^2}$:] $\vali[n+1] = n^2+\eps, \vali = 1,$ for all $i\in[n]$. The total payment is $\sum_{i\in[n+1]}\payi=n^2+\eps$.
\item[with probability $\frac{1}{n} - \frac{1}{n^2}$:] $\vali[n+1] = n^2+\eps$,  draw independently $\vali\sim (\disti ~|~\vali<1)$ for $i\in[n]$. The total payment is $\sum_{i\in[n+1]}\payi=n+\eps$.
\item[with probability $1 - \frac{1}{n}$:] draw  
$\vali[n+1]\sim(\disti[n+1] ~|~\vali[n+1]<n^2+\eps)$; choose one bidder $i^*\in[n]$ uniformly at random and let $\vali[i^*]=1$; draw independently $\vali\sim (\disti ~|~\vali<1)$
for $i\in[n], i\ne i^*$. The total payment is $\sum_{i\in[n+1]}\payi=1$.
\end{description}
We verify that $\dists$ is a pairwise independent distribution with the given set of marginals. Indeed, the last bidder has value $\vali[n+1]=n^2+\eps$ in the first two cases with the total probability of $1/n^2+(1/n-1/n^2)=1/n$; when this happens ($\vali[n+1]=n^2+\eps$) the distribution of any bidder $i\in[n]$ is exactly $\ER[\frac{1}{n},1]$ ($\vali=1$ with probability $\frac{1/n^2}{1/n^2+(1/n-1/n^2)}=\frac{1}{n}$, and otherwise is distributed as $\vali\sim (\disti ~|~\vali<1)$). Furthermore, 
the last bidder has value $\vali[n+1]< n^2+\eps$ only in the third case and conditioned on that $\vali[n+1]\sim\disti[n+1]$; $\vali=1$ with probability $\Prx{i=i^*}=\frac{1}{n}$ and otherwise $\vali\sim (\disti ~|~\vali<1)$ for each fixed $\vali[n+1]<n^2+\eps$ and $i\in[n]$. Finally, 
each of the two bidders $i,j\in[n]$ has $\ER[\frac{1}{n},1]$ marginal distribution. 
They both have $\vali=\vali[j]=1$ only in the first case with probability $1/n^2=1/n\cdot 1/n$; also $\vali\sim (\disti | \vali<1, \vali[j]=1)$, as the event $\{\vali<1,\vali[j]=1\}$ only happens in the third case with probability 
$(1-1/n)\cdot\prob{i^*=j}=\frac{1-1/n}{n}$.
The total revenue $\myerson(\dists)=\frac{1}{n^2}\cdot n^2 + \InParentheses{\frac{1}{n}-\frac{1}{n^2}}\cdot n+\InParentheses{1-\frac{1}{n}}\cdot 1 < 3$, which concludes the proof.  
\end{proof}

In contrast, if we implement a second price auction with an anonymous reserve price $n^2 + \eps$. For both pairwise and mutually independent priors, the revenue is $n + \frac{\eps}{n}$ which means that AR is pairwise-robust for these marginal distributions.

%% file: sec/3wi.tex
The example from Section~\ref{sec: nega} shows that Myerson's mechanism is not pairwise-robust for regular marginal distributions. However, as we show in this section, it is $3$-wise-robust for any (even irregular) marginals $(\disti)_{i\in[n]}$. We still need to be careful when defining Myerson's mechanism, since it might not be unique for distributions with point masses and/or irregular marginals, because of the different tie-breaking rules. In the theorem below, we make certain assumptions on the tie-breaking rule in these cases. 
\begin{theorem}
\label{thm:3wise_robust}
Myerson's mechanism is $3$-wise-robust.
\end{theorem}
\begin{proof} We give the proof for regular marginals to simplify the presentation (the same proof also works for the irregular cases, but all virtual values $\virti(\vali)$ would have to be changed to ironed virtual values $\virtiri(\vali)$).
We first consider an ex-ante $1/2$-relaxation of Myerson's mechanism, in which at most $1/2$ of the item is sold in expectation. This relaxation is at least a $1/2$-approximation to Myerson's revenue $\myerson(\distind)$ on the mutually independent prior. 
The optimal ex-ante mechanism sells the item separately to each bidder $i\in[n]$ at the prices $\payi$ that give the same virtual values $\virti(\payi)=\nu$ across all agents.   
We define\footnote{Normally, $\ttau$ is a unique threshold such that $\sum_{i\in[n]}\Prx{\virti(\vali)\ge\ttau}=\frac{1}{2}$, but some distributions $\disti$ may have atoms at $\ttau$. In this case
$\sum_{i\in[n]}\Prx{\virti(\vali)\ge\ttau}>\frac{1}{2}$. 
Then we use a tie-breaking rule that with certain probability gives 
an agent $i$ with atom at $\ttau$ 
and $\virti(\vali)=\ttau$ the lowest priority in the tie-breaking as if 
their virtual value 
$\virti(\vali)=\ttau-\varepsilon$. In this way we ensure that $\sum_{i\in[n]}\Prx{\virti(\vali)\ge\ttau}=\frac{1}{2}$.} the virtual value level of ex-ante $1/2$-relaxation as
\[
\ttau\eqdef \inf\left\{\nu~\middle|~\sum\limits_{i\in[n]}\Quani\left(\virti^{-1}(\nu)\right)\le \frac{1}{2}\right\}\quad\quad
\text{with tie-breaking }
\sum\limits_{i\in[n]}\Prx{\virti(\vali)\ge\ttau}=\frac{1}{2}.
\]
Note that the level $\ttau$ might be less than $0$, when the ex-ante probability $\sum_{i\in[n]}\Quani\left(\mri\right)$ of selling item at monopoly reserves $\mri$ (corresponds to $0$ level of virtual values $\virti(\mri)=0$) is less than $\frac{1}{2}$.
This is an easy case, which we consider next.

\textbf{Case 1:} $\ttau \leq 0$. In this case,
if $\vali\ge\mri$ for a bidder $i\in[n]$, then
with probability at least half $i$ is the only bidder beating her monopoly reserve for any prior $\dists\in\distpwi$. 
Indeed, consider two events $\event_i\eqdef \{\vali\ge \mri\}$ and $\bigwedge_{j\ne i}\overline{\event_{j}} = \{\forall j\neq i, \vali[j] < \mri[j]\}$ for each $i\in[n]$. Then we have the following bound for any pairwise independent prior $\dists\in\distpwi$.
\begin{multline*}
    \Exlong[\vals\sim\dists]{\myerson(\vals)} \ge
    \sum_{i\in[n]}\Ex[\vals]{\myerson(\vals)\cdot
    \ind{\event_i\wedge\bigwedge\nolimits_{j\ne i}\overline{\event_{j}}}}
    =\sum_{i\in[n]}\mri\cdot\Prx{\event_i}\cdot
    \Prx{\bigwedge\nolimits_{j\ne i}\overline{\event_j}~\left|\vphantom{\bigcup\nolimits_{j\ne i}}\right.~\event_i}\\
    =\sum_{i\in[n]}\mri\cdot\Quani(\mri)\cdot
    \Prx{\overline{\bigvee\nolimits_{j\ne i}\event_j}~\left|\vphantom{\bigcup\nolimits_{j\ne i}}\right.~\event_i}
    \ge
    \sum_{i\in[n]}\mri\cdot\Quani(\mri)\cdot\left(1-\sum_{j\ne i}\Prx{\event_j~|~\event_i}\right)
    \\
    =
    \sum_{i\in[n]}\mri\cdot\Quani(\mri)\cdot\left(1-\sum_{j\ne i}\Prx{\event_j}\right)
    \ge\frac{1}{2}\sum_{i\in[n]}\mri\cdot\Quani(\mri)\ge\frac{1}{2} \Ex[\vals\sim\distind]{\myerson(\vals)},
\end{multline*}
where the first inequality holds, since the events $\event_i\wedge\bigwedge\nolimits_{j\ne i}\overline{\event_{j}}$ are disjoint for all $i\in[n]$; the first equality holds, since when the event $\event_i\wedge\bigwedge\nolimits_{j\ne i}\overline{\event_{j}}$ occurs with probability $\Prx{\event_i}\cdot
    \Prx{\bigwedge\nolimits_{j\ne i}\overline{\event_j}~\left|\vphantom{\bigwedge\nolimits_{j\ne i}}\right.~\event_i}$, agent $i$ pays her monopoly reserve $\mri$; in the second equality we simply used that $\Prx{\event_i}=\Quani(\mri)$ and that $\bigwedge\nolimits_{j\ne i}\overline{\event_{j}}=\overline{\bigvee\nolimits_{j\ne i}\event_j}$; 
    the second inequality is simply the union bound for events $\{\event_j\}_{j\ne i}$ and $\overline{\bigvee\nolimits_{j\ne i}\event_j}$ under condition of $\event_i$; the third equality follows from pairwise independence of the distribution $\dists$; the third inequality holds as $\sum_{i\in[n]}\Prx{\event_i}\le\frac{1}{2}$ when $\ttau \leq 0$; the last inequality holds as $\sum_{i}\mri\cdot\Quani(\mri)$ is the revenue of the ex-ante relaxation with unlimited supply, which is at least the revenue of any single-item auction such as $\Ex{\myerson(\vals)}$. 

Thus, Myerson's mechanism is pairwise-robust with a constant $2$. We consider next the non-degenerate case $\ttau >0$, when we actually need to use that $\dists\in\disttwi$ rather than $\dists\in\distpwi$.
 
\textbf{Case 2:} $\ttau > 0$. In this case, we consider the events $\event_i\wedge\event_j=\{\virti(\vali),\virti[j](\vali[j])\ge\ttau\}$ and $\bigwedge_{k\ne i,j}\overline{\event_k}=\{\forall k\ne i,j,\virti[k](\vali[k])<\ttau\}$ for each pair of bidders $i\ne j$, where $\event_i\eqdef\{\virti(\vali)\ge\ttau\}$. 
Then for any 3-wise independent prior $\dists\in\disttwi$ and any realization of $\vali,\vali[j]$,
all remaining bidders $k\ne i,j$ have virtual values strictly below $\ttau$ with constant probability.
\begin{claim}
    \label{cl:ex_ante_half}
    For any $\dists\in\disttwi$ and any fixed  $\vali$ and $\vali[j]$ of bidders $i\ne j$ we have
    \begin{equation}
    \label{eq:ex_ante_half}
    \Prx[\vals\sim\dists]{\bigwedge\nolimits_{k\ne i,j}\overline{\event_k} \left|\vphantom{\bigwedge\nolimits_{k}}\right. \vali,\vali[j]}\ge\frac{1}{2}~.
    \end{equation}
\end{claim}
\begin{proof} By the union bound for the events $\{\event_k\}_{k\ne i,j}$ and $\overline{\bigvee_{k\ne i,j}\event_k}$ conditioned on $\vali,\vali[j]$ we get
\begin{multline*}
    \Prx{\bigwedge\nolimits_{k\ne i,j}\overline{\event_k} \left|\vphantom{\bigwedge\nolimits_{k}}\right. \vali,\vali[j]}=
    \Prx{\overline{\bigvee\nolimits_{k\ne i,j}\event_k} \left|\vphantom{\bigwedge\nolimits_{k}}\right. \vali,\vali[j]}
    \ge
    1-\sum_{k\ne i,j}\Prx{\event_k~|~\vali,\vali[j]}\\
    =
    1-\sum_{k\ne i,j}\Prx{\event_k}
    =1-\sum_{k\ne i,j}\Prx{\virti[k](\vali[k])\ge\ttau}\ge\frac{1}{2},
\end{multline*}    
where in the second equality we used $3$-wise independence of $\dists$; 
the last inequality follows from the definition of $\ttau$. 
\end{proof}
We use disjoint events $\left\{\event_i\wedge\event_j\wedge\bigwedge_{k\ne i,j}\overline{\event_k}\right\}_{i\ne j}$ to get the lower bound~\eqref{eq: 3wi} on the expected revenue of Myerson's mechanism with any $3$-wise independent prior $\dists\sim\disttwi$.
To shorten notations, we denote 
for each bidder $i\in[n]$ her ex-ante value threshold as $\exvi\eqdef\virt^{-1}(\ttau)$, the probability $\Prx{\vali\ge\exvi}=\Prx{\event_i}$ as $\qi$, and her revenue in the ex-ante $1/2$-relaxation at price $\exvi$ as $\revi\eqdef\exvi\cdot\Prx{\vali\ge\exvi}=\exvi\cdot\qi$. We have
\begin{multline}    
\label{eq: 3wi}
\Exlong[\vals\sim\dists]{\myerson(\vals)} \ge 
\sum\limits_{i\neq j}\Exlong[\vals]{\myerson(\vals)\cdot 
\ind{\event_i\wedge\event_j\wedge\bigwedge_{k\ne i,j}\overline{\event_k}} }
=\sum\limits_{i\neq j}\Prx{\event_i\wedge\event_j}\\
\cdot\Exlong[\vals]{\myerson(\vals)\cdot
\ind{\bigwedge\nolimits_{k\ne i,j}\overline{\event_k}} \left|\vphantom{\bigwedge\nolimits_{k}}\right. \event_i\wedge\event_j
}
=\sum\limits_{i\neq j}\qi\cdot\qi[j]\cdot
\Exlong[\substack{\disti|\vali\ge\exvi\\
\disti[j]|\vali[j]\ge\exvi[j]}]{\myerson(\vali,\vali[j])\cdot\Prx{\bigwedge_{k\ne i,j}\overline{\event_k}\left|\vphantom{\bigwedge_{k}}\right.\vali,\vali[j]}}
\\
\ge\sum\limits_{i\neq j}\qi\cdot\qi[j]\cdot\frac{1}{2}\cdot
\Exlong[\substack{\disti|\vali\ge\exvi\\
\disti[j]|\vali[j]\ge\exvi[j]}]{\myerson(\vali,\vali[j])}
\ge
\sum\limits_{i\neq j}\qi\cdot\qi[j]\cdot\frac{1}{2}\cdot\frac{1}{2}\left(\exvi+\exvi[j]\right)
\\
=\frac{1}{4}\sum\limits_{i\neq j} \bigg(\revi\cdot\qi[j]+\revi[j]\cdot\qi\bigg)
=\frac{1}{4}\sum\limits_{i\in[n]}\revi\cdot\sum_{j\ne i}\qi[j],~~~~~~
\end{multline}
where 
the second equality holds as the
payment of Myerson's mechanism in the event $\event_i\wedge\event_j\wedge\bigwedge_{k\ne i,j}\overline{\event_k}$ is completely determined by values $\vali,\vali[j]$; the second inequality holds by Claim~\ref{cl:ex_ante_half}; the third inequality holds, since by running the optimal single-item auction twice we get at least as much revenue as selling $2$ copies of the item separately to bidders $i$ and $j$ at respective prices $\exvi$ and $\exvi[j]$; the third equality holds by the definition of $\revi$ and $\revi[j]$.

On the other hand, since ex-ante $1/2$-relaxation is at least a $1/2$-approximation to the standard ex-ante relaxation, which in turn is an upper bound on the Myerson's revenue with the mutually independent prior $\distind$, we have
\begin{equation}
\label{eq: 3witau}
2\sum\limits_{i\in[n]}\revi \ge \Ex[\vals\sim\distind]{\myerson(\vals)}.
\end{equation}
With the bounds \eqref{eq: 3wi} and \eqref{eq: 3witau}, the proof of Theorem~\ref{thm:3wise_robust} is almost complete. We only need to consider a few straightforward cases to relate the lower bound \eqref{eq: 3wi} on $\Ex[\vals\sim\dists]{\myerson(\vals)}$ with the upper bound \eqref{eq: 3witau} on $\Ex[\vals\sim\distind]{\myerson(\vals)}$.
Without loss of generality let us assume that $\qi[1]=\Prx{\virti[1](\vali[1])\ge\ttau}=\max_{i\in[n]}\qi$ in the reminder of the proof.
\textbf{Case 2.1:} 
$\qi[1]=\max_{i}\qi\le\frac{1}{4}$.
According to \eqref{eq: 3wi}, for any $\dist\in \disttwi$ we have
\begin{equation*}
\Ex[\vals\sim\dist]{\myerson(\vals)} \ge   
\frac{1}{4}\sum\limits_{i\in[n]}\revi\cdot\sum_{j\ne i}\qi[j]
\ge \frac{1}{16}\sum\limits_{i\in[n]}\revi 
\ge \frac{1}{32} 
\Ex[\vals\sim\distind]{\myerson(\vals)},
\end{equation*}
where the second inequality is due to $\sum_{j\neq i}\qi[j]\ge\left(\sum_{j}\qi[j]\right)-\qi[1]\ge\frac{1}{2}-\frac{1}{4}$ for any $i\in[n]$; the third inequality holds by \eqref{eq: 3witau}.

\textbf{Case 2.2:} 
$\qi[1]\ge \frac{1}{4}$. I.e., the first agent contributes the most to the ex-ante probability $\sum_{i\in[n]}\qi$. We consider two subcases depending on her contribution $\revi[1]$ to the total revenue of the ex-ante $1/2$-relaxation $\sum_{i\in[n]}\revi$.   

\textbf{Case a:} $\revi[1]\le \frac{1}{2} \sum_{i\in[n]}\revi.$
By \eqref{eq: 3wi} we have
\begin{equation*}
\Ex[\vals\sim\dists]{\myerson(\vals)}\ge   
\frac{1}{4}\sum\limits_{i\in[n]}\revi\cdot\sum_{j\ne i}\qi[j]
\ge \frac{1}{16}\sum\limits_{i\neq 1}\revi
\ge \frac{1}{32}\sum\limits_{i\in[n]}\revi 
\ge \frac{1}{64} \Ex[\vals\sim\distind]{\myerson(\vals)},
\end{equation*}
where the second inequality is due to 
$\sum_{j\neq i}\qi[j]\ge\qi[1]\ge\frac{1}{4}$ for $i\ne 1$;
the third inequality holds since $\revi[1]$ is less than a half of the total revenue $\sum_{i\in[n]}\revi$; \eqref{eq: 3witau} gives the last inequality.

\textbf{Case b:} $\revi[1]> \frac{1}{2} \sum_{i\in[n]}\revi.$
In this case, we consider situations in which the first agent wins and pays at least her monopoly reserve price $\mri[1]$. This happens when $\event_1\wedge\bigwedge_{i\ne 1}\overline{\event_i}=\{\virti[1](\vali[1])\ge\ttau,\forall i\ne 1, ~\virti(\vali)<\ttau\}$. Then, for any pairwise independent prior $\dists\in\distpwi$ we have 
\begin{multline*}
    \Exlong[\vals\sim\dists]{\myerson(\vals)}\ge
    \Exlong[\vals]{\myerson(\vals)\cdot 
    \ind{\event_1\wedge\bigwedge\nolimits_{i\ne 1}\overline{\event_i}} }
    =\Prx{\event_1}\cdot\Exlong[\vals]{\myerson(\vals)\cdot 
    \ind{\bigwedge\nolimits_{i\ne 1}\overline{\event_i}} \left|\vphantom{\bigwedge\nolimits_{i}}\right. \event_1}\\
    \ge
    \qi[1]\cdot\Exlong[\vals]{\mri[1]\cdot\ind{\bigwedge\nolimits_{i\ne 1}\overline{\event_i}} \left|\vphantom{\bigwedge\nolimits_{i}}\right. \event_1}
    =
    \qi[1]\cdot\mri[1]\cdot\Prx{\overline{\bigvee\nolimits_{i\ne 1}\event_i} \left|\vphantom{\bigvee\nolimits_{i}}\right. \event_1}
    \ge
    \qi[1]\cdot\mri[1]\cdot\left(1-\sum_{i\ne 1}\Prx{\event_i ~|~\event_1}\right)
    \\
    =\qi[1]\cdot\mri[1]\cdot
    \left(1+\qi[1]-\sum_{i\in[n]}\qi\right)\ge\frac{1}{4}\cdot\revi[1]\cdot
    \frac{3}{4}\ge\frac{3}{32}\sum_{i\in[n]}\revi
    \ge
    \frac{3}{64}\Exlong[\vals\sim\distind]{\myerson(\vals)},
\end{multline*}
where to get the second inequality we observe that 
the first agent wins in the events $\event_1$ and 
$\bigwedge_{i\ne 1}\overline{\event_i}$ and has to pay at least her reservation price 
$\mri[1]$; the third inequality is just a union bound for the events $(\event_i)_{i\ne 1}$ and 
$\overline{\bigvee_{i\ne 1}\event_i}$; we used pairwise independence and that $\qi=\Prx{\event_i}$ to get the third equality; in the forth inequality we used that monopoly price
$\mri[1]\ge\mri[1]\cdot\Prx{\virti[1](\vali[1])\ge 0}\ge\revi[1]$ and the conditions 
$\qi[1]\ge\frac{1}{4}$ and  $\sum_{i\in[n]}\qi=\frac{1}{2}$;
in the fifth and sixth inequalities, we used case (b) condition and \eqref{eq: 3witau}. Thus we conclude the proof of Theorem~\ref{thm:3wise_robust}.
%
\end{proof}

%% file: sec/bounds.tex
Our next goal is to show that Myerson's auction, which is equivalent to the second price auction with an anonymous reserve, is not over-optimizing for the mutually independent prior, unlike Myerson's mechanism. I.e., the revenue $\AR\left(\reserve\right)$ on any given pairwise independent prior $\dists\in\distpwi$ is within a constant factor from its revenue $\AR^{ind}\left (\reserve\right)$ on the independent prior $\distind$. To this end, we express these revenues as integrals of respective probabilities that at least one or at least two bidders exceed a threshold value $\tau$.
\be
\AR\left (\reserve\right ) = \reserve\cdot \Qone\left (\reserve\right ) + \int_{\reserve}^{+\infty}\Qtwo\left (\tau\right )\d \tau,
\quad\quad
\AR^{ind}\left (\reserve\right ) = \reserve\cdot \Qoneind\left (\reserve\right ) + \int_{\reserve}^{+\infty}\Qtwoind\left (\tau\right )\d \tau,\label{eq: pwi_rev}
\ee
where the respective probabilities for the pairwise and mutually independent distributions $\dists$ and $\distind$ are
\begin{align*}
\Qone(\tau) \eqdef \Prx[\vals\sim\dists]{|\{i:\vali\ge \tau\}| \ge 1},&
\quad\quad 
\Qtwo(\tau) \eqdef \Prx[\vals\sim\dists]{|\{i:\vali\ge \tau\}| \ge 2};\\
\Qoneind(\tau) \eqdef \Prx[\vals\sim\distind]{|\{i:\vali\ge \tau\}| \ge 1},&
\quad\quad 
\Qtwoind(\tau) \eqdef \Prx[\vals\sim\distind]{|\{i:\vali\ge \tau\}| \ge 2}.
\end{align*}
The integral forms in~\eqref{eq: pwi_rev} allow us to compare respective terms $\Qone(\reserve)$ with $\Qoneind(\reserve)$ and 
$\Qtwo(\tau)$ with $\Qtwoind(\tau)$ for each $\tau\in[\reserve,\infty)$. For the former pair, the probability $\Qone(\reserve)$ is at least a constant fraction of $\Qoneind(\reserve)$. In fact, it was shown by Caragiannis et al.~\cite{caragiannis2022relaxing} to hold for any $\Qone(\tau),\Qoneind(\tau)$ and not only for $\tau=\reserve$. 
\begin{lemma}[\cite{caragiannis2022relaxing}]
\label{lem: wine}
For any $\dists\in\distpwi$  and any threshold $\tau\ge 0$,
$ 1.299\cdot \Qone\left (\tau\right ) \ge \Qoneind\left (\tau\right ).
$
\end{lemma}
This lemma is the main tool to prove the robustness of several mechanisms and algorithms such as sequential posted pricing and prophet inequality under pairwise independent distribution  in~\cite{caragiannis2022relaxing}. Now, if we could give a similar lower bound on $\Qtwo(\tau)$ via $\Qtwoind(\tau)$ for any $\tau\in[\reserve,\infty)$, then we would immediately get the desired approximation guarantee for $\AR(\reserve)$ via $\AR^{ind}(\reserve)$. 
As it turns out, the probabilities $\Qtwo(\tau)$ and $\Qtwoind(\tau)$ do not behave as nicely as $\Qone(\tau)$ and $\Qoneind(\tau)$. We give below a counter-example showing that the gap between $\Qtwo(\tau)$ and $\Qtwoind(\tau)$ for some $\tau$ can be arbitrarily large even for the case of identical marginals.

\paragraph{Example.} 
There are $n + 1$ bidders, each bidder's values distributed uniformly $\uniform[0, 1]$ on $[0,1]$ interval. We set $\tau = \frac{n - 1}{n}$. Then, 
$
    \Qtwoind\left(\frac{n - 1}{n}\right) = 1 - \left(1 - \frac{1}{n}\right)^{n + 1} - (n + 1)\cdot\frac{1}{n}\cdot\left(1 - \frac{1}{n}\right)^n\xrightarrow[n\to\infty]{ } 1 - \frac{2}{e}.
$
We construct a pairwise independent joint distribution $\dists\in\distpwi$ with $\Qtwo(\frac{n-1}{n})=o(1)$ as follows.
\begin{description}
\item[with probability $\frac{1}{n^2}$:] Draw independently $\vali\sim \uniform[\frac{n - 1}{n}, 1]$ for $i\in [n + 1]$.
\item[with probability $1 - \frac{1}{n^2}$:] Choose a bidder $i^*\in [n + 1]$ uniformly at random and let $\vali[i^*]\sim \uniform[\frac{n - 1}{n}, 1]$; draw independently $\vali\sim\uniform[0, \frac{n - 1}{n}]$ for $i\neq i^*$.
\end{description}
To see that this distribution $\dists$ is pairwise independent, we first note that the marginal distribution of each bidder is $\uniform[0,1]$. Indeed, the probability that value of bidder $i$ is larger than $\frac{n - 1}{n}$ is $\frac{1}{n^2} + \frac{1}{n + 1}\cdot (1 - \frac{1}{n^2}) = \frac{1}{n}$ and given that $\vali\ge\frac{n-1}{n}$ or that $\vali\le\frac{n-1}{n}$ the distribution is uniform.
Second, for any given pair of bidders, the probability that their values are larger than $\frac{n - 1}{n}$ is $\frac{1}{n^2}$. Similarly,
the probability that the value of one bidder is larger than $\frac{n - 1}{n}$ and the other value is smaller than $\frac{n - 1}{n}$ is $\frac{1}{n + 1}\cdot (1 - \frac{1}{n^2}) = \frac{n - 1}{n^2}$;
and the probability that both values are smaller than $\frac{n - 1}{n}$ is $\frac{n - 1}{n + 1}\cdot (1 - \frac{1}{n^2}) = \frac{(n - 1)^2}{n^2}$. Thus $\dists$ is a pairwise independent distribution.
Then $\Qtwo\left(\frac{n-1}{n}\right) = \frac{1}{n^2}$, which quickly converges to $0$ when $n$ goes to $+\infty$. In the meantime $\Qtwoind\left(\frac{n-1}{n}\right)$ is constant.

\quad

The example above demonstrates that unlike Lemma~\ref{lem: wine}, the multiplicative gap between $\Qtwo(\tau)$ and $\Qtwoind(\tau)$ may be unbounded. I.e., one cannot simply compare similar terms in the integral representations~\eqref{eq: pwi_rev} to establish the desired approximation result. On the positive side, the counter-example only works for some specific $\tau$: $\tau$'s such that the expected number of bidders $\sum_{i = 1}^n\Quani(\tau)$ above the threshold $\tau$ is not much larger than $1$. Indeed, for each fixed $\tau$ we have simple 
Bernoulli random variables $X_i$ corresponding to the event $\{\vali\ge \tau\}$, where each $X_i=1$ with probability $\Quani(\tau)$ and $X_i=0$ otherwise. Then $\Qone(\tau)=\Prx{\sum_{i\in[n]}X_i\ge 1}$ and $\Qtwo(\tau)=\Prx{\sum_{i\in[n]}X_i\ge 2}$ are respectively the probabilities that at least one/two of $X_i$ are ones. As it turns out, 
there are well-established bounds on the probability of at least one/two events occurring at the same time under pairwise independent distributions. The crucial parameter for such bounds is the expected number of occurring events $\momone\eqdef\sum_{i = 1}^n\Quani(\tau)$.

\begin{lemma}[\cite{boros1989closed}]
\label{lem: ziyong}
Let $\sum_{i = 1}^n\Quani(\tau) = \momone$.
For any pairwise independent distribution $\dists$
\begin{align}
\label{eq: ziyong1}
    \forall \momone\ge 0
    && \Qone(\tau)\geq \frac{2m_1\momone - \momone^2}{m_1\left (m_1 + 1\right )}\eqdef \Lowerboundone\left (\momone\right), 
    &\quad\quad\text{where }m_1 = \lfloor \momone + 1\rfloor.\\
\label{eq: ziyong2}
    \forall \momone>1
    && \Qtwo\left (\tau\right )\geq \frac{2m_2\left (\momone - 1\right ) - \momone^2}{m_2\left (m_2 - 1\right )}\eqdef \Lowerboundtwo\left (\momone\right),
    &\quad\quad\text{where }m_2 = \left\lfloor \frac{\momone^2}{\momone - 1}\right\rfloor.    
\end{align}
\end{lemma}
In fact, the inequalities \eqref{eq: ziyong1} and \eqref{eq: ziyong2} are valid for any positive integers $m_1, m_2$, but the choices $m_1=\lfloor \momone + 1\rfloor$ and $m_2 = \left\lfloor \frac{\momone^2}{\momone - 1}\right\rfloor$ simply maximize the respective $\Lowerboundone(s)$ and $\Lowerboundtwo(s)$. It is also not hard to verify that $\Lowerboundone(\momone)$ and $\Lowerboundtwo(\momone)$ are monotonically increasing in $\momone$.

Note that the bound~\eqref{eq: ziyong2} is useful only for the thresholds $\tau$ with $\momone>1$. But, the main contribution to the integral terms of~\eqref{eq: pwi_rev} may be for $\tau$ with $\momone=\sum_{i = 1}^n\Quani(\tau)$ is smaller than $1$. To get the desired approximation, we will have to compare the respective integral parts of~\eqref{eq: pwi_rev} for $\momone\le 1$ with $\Qone(\reserve)$. 

In what follows we make the necessary comparisons: first, for an easier case of identical marginals $\dist=\disti$ for all $i\in[n]$ and specific reserve price $\reserve$ equal to the monopoly reserve $\reserve=\mr$ of $\dist$; next, for a more technically involved case of nonidentical marginals $(\disti)_{i\in[n]}$ and arbitrary reserve price $\reserve$. Before proceeding to the case of identical marginals, we present another simple bound on $\Qone(\tau)$ from \cite{caragiannis2022relaxing} that depends on $\momone$.
\begin{lemma}[\cite{caragiannis2022relaxing}]
\label{lem: wine2}
Let $\sum_{i = 1}^n\Quani(\tau) = \momone$ for $\tau\ge 0$. Then $\Qone(\tau)\geq \frac{\momone}{\momone + 1}$ for any $\dists\in\distpwi$.
\end{lemma}

%% file: sec/iden.tex
We consider in this section a special case of identical marginal distributions $\disti=\dist$ for all $i\in[n]$. 
The Myerson's auction in this case is the same as $\AR(\mr)$ the second price auction with monopoly reserve $\mr$ (the virtual value $\virt(\mr)=0$), when the prior $\dist$ is a regular distribution. 
\begin{theorem}
\label{thm: myersoniid}
SPA with monopoly reserve $\mr$ is pairwise-robust for any identical regular marginals $\dist$ with a constant $c< 2.63$.
\end{theorem}
\begin{proof}

We denote the revenue of the optimal Myerson auction under mutually independent distribution $\distind$ as $\opt$ and derive an upper bound on $\opt$ similar to a widely used Ex-ante relaxation. 

Let $\ww$ be the threshold such that $\Quan(\ww) = \frac{1}{n}$ (if $\ww$ does not exist, we define $\ww\eqdef \sup\{\val~|~\Quan(\val)\ge \frac{1}{n}\}$), then the Ex-ante relaxation gives an upper bound $\opt\leq \ww$. 
The following lemma is essentially the well-known Ex-ante relaxation upper bound on $\opt$.
\begin{lemma}[Ex-ante relaxation]
\label{lem: exante}
For any regular distribution $\dist$, let 
$\valo\in [\mr,\ww)$ when $\mr\le\ww$, or $\valo=\mr$ when $\ww<\mr$. Then, 
\begin{equation}
    \valo \geq \frac{\opt}{n\cdot\Quan(\valo)}.
\end{equation}
\end{lemma}
\begin{proof} 
We need to prove that $\valo\cdot n\cdot\Quan(\valo)\ge\opt$. First, we consider the case $\mr\le\ww$. 
Consider the optimal mechanism for selling ex-ante (in expectation) $n\cdot\Quan(\valo)$ items. Since 
$n\cdot\Quan(\valo)\ge 1$, this mechanism is a relaxation of the optimal auction and gets at least as 
much revenue. On the other hand, as $\dist$ is regular and $\valo\ge\mr$, the optimal ex-ante mechanism 
sells to each bidder at price $\valo$ and thus has revenue $n\cdot\valo\cdot\Quan(\valo)\ge\opt$.
In the case $\ww<\mr$, consider the optimal mechanism for selling items in unlimited supply. It offers the monopoly reserve $\mr=\valo$ to each agent and thus 
gets revenue 
$n\cdot\valo\cdot\Quan(\valo)$. Its revenue is an obvious upper bound on $\opt$. Thus $n\cdot\valo\cdot\Quan(\valo)\ge\opt$.
\end{proof}

We are now ready to finish the proof of Theorem~\ref{thm: myersoniid}.
The first useful observation is that for each $\valo < \ww$ (i.e., $\Quan(\valo)> \frac{1}{n}$) there is a positive probability that at least two bidders have values larger than $\valo$ by Lemma~\ref{lem: ziyong}.
We also need to consider the relation between the monopoly reserve $\mr$ and $\ww$. To this end, let us define $\kk\eqdef \frac{1}{n\cdot\Quan\left (\mr\right )}$. We consider two cases. 
First, if $\mr \geq \ww$, i.e. $\kk \ge 1$, then 
\begin{equation*}
\Rev\left (\dists, \gM\right ) \geq \frac{\reserve^*}{\kk + 1} \geq \frac{\kk\cdot \opt}{\kk + 1} \geq \frac{1}{2}\cdot \opt,
\end{equation*}
where the first inequality holds by Lemma~\ref{lem: wine2},
as we sell the item with probability $\Qone(\mr)\ge\frac{n\cdot\Quan(\mr)}{1+n\cdot\Quan(\mr)}=\frac{1}{1+\kk}$ with the payment of at least $\mr$; the second inequality holds by Lemma~\ref{lem: exante} for $\valo = \mr$; the last inequality holds, as $\kk \ge 1$.

Second, if $\mr < \ww$. Let us consider the payment $\Payment=\sum_{i\in[n]}\payi(\vals)$ in Myerson's auction as a random variable for pair-wise independent prior $\vals\sim\dists$. To compare the expected value of $\Payment$ to $\opt$, we would like to apply Lemma~\ref{lem: exante} for each realization of $\Payment$ that is at least $\mr$. As Lemma~\ref{lem: exante} only applies to $\Payment \in  [\reserve^*, \ww)$, we round down all realization of $\Payment\ge\ww$ to $\ww$ by defining another random variable $\Payment'\eqdef\min\{\Payment, \ww\}$. Specifically, we get  
\begin{equation*}
\label{eq: thm1}
\Ex[\vals\sim\dists]{\Rev(\vals)} = \Ex{\Payment} \ge \Ex{\Payment'}\geq \Ex{\frac{\ind{\Payment'\ge\mr}\cdot \opt}{n\cdot\Quan\left (\Payment'\right )}}.
\end{equation*}
Hence, we get the following lower bound on the ratio of $\Ex[\vals\sim\dists]{\Rev(\vals)}$ and $\opt$.

\begin{multline}
\frac{\Ex[\vals\sim\dists]{\Rev(\vals)}}{\opt} \ge \Ex{\frac{\ind{\Payment'\ge\mr}}{n\cdot\Quan\left (\Payment'\right )}}
 = \kk\cdot \prob{\frac{1}{n\cdot\Quan\left (\Payment'\right )}\geq \kk} + \int_\kk^1\prob{\frac{1}{n\cdot\Quan\left (\Payment'\right )}\geq u}\d u \\
 = \kk\cdot \prob{\Payment' \geq \mr}+ \int_\kk^1\prob{\Payment'\geq \Quan^{-1}\left(\frac{1}{n\cdot u}\right)}\d u,
 \label{eq: iidratio}
\end{multline}
where the first equality holds as (a) $\Ex{Y}=\int_{0}^{\infty}\Prx{Y\ge u}\d u$ for any random variable $Y\ge 0$ and (b) for   $Y=\frac{\ind{\Payment'\ge\mr}}{n\cdot\Quan\left (\Payment'\right )}$ we have $Y\in[\kk,1]\cup\{0\}$, as $\Payment'=0$ or  
$\mr\leq \Payment'\leq \ww$, and in the later case $\kk=\frac{1}{n\cdot\Quan\left (\mr\right )}\le \frac{1}{n\cdot\Quan\left (\Payment'\right )}\le\frac{1}{n\cdot\Quan\left (\ww\right )}\le 1$ (usually $n\cdot\Quan\left (\ww\right )=1$, but, if $\dist$ has an atom, we might have instead $\frac{1}{n\cdot\Quan\left (\ww\right )}\le\lim_{\eps\to +0}\frac{1}{n\cdot\Quan(\ww+\eps)}=1$).

The event $\{\Payment' \ge \mr\}$ in \eqref{eq: iidratio} is equivalent to at least one bidder having value larger than $\mr$. I.e., we may use $\Lowerboundone\left(\mr\right )$ from Lemma~\ref{lem: ziyong} as a lower bound for $\prob{\Payment' \geq \mr}$ with $\sum_{i = 1}^n\Quani(\mr) = n\cdot \Quan(\mr) = \frac{1}{\kk}$.
The event $\{\Payment' \ge \tau\eqdef\Quan^{-1}(\frac{1}{n\cdot u})\}$ for $u\in[\kk, 1]$ is equivalent to at least two bidders having values greater than or equal to $\tau$. We use $\Lowerboundtwo\left (\tau\right )$ from Lemma~\ref{lem: ziyong} as a lower bound for $\prob{\Payment'\geq \tau}$ with  $\sum_{i = 1}^n\Quani(\tau)=n\cdot\Quan(\tau) =\frac{1}{u}$. Hence, after applying Lemma~\ref{lem: ziyong} to \eqref{eq: iidratio} we get
\begin{multline}
\label{eq:minimization_beta}    
\frac{\Ex[\vals\sim\dists]{\Rev(\vals)}}{\opt} \ge \kk\cdot \Lowerboundone\left (\frac{1}{\kk}\right ) + \int_\kk^1\Lowerboundtwo\left (\frac{1}{u}\right )\d u
\\
\ge \min_{\beta\in[0, 1]}\left\{\beta\cdot \Lowerboundone\left (\frac{1}{\beta}\right ) + \int_\beta^1\Lowerboundtwo\left (\frac{1}{u}\right )\d u\right\}.
\end{multline}

We find quantity in the RHS of the lower bound~\eqref{eq:minimization_beta} by numerically optimizing for the worst parameter $\beta\in[0,1]$ (note that $\Lowerboundone(x)$ and $\Lowerboundtwo(x)$ are two simple explicit functions).
Figure~\ref{fig: proof_ratio} (a) shows the curve of $\Lowerboundone\left (\frac{1}{v}\right )$ and $\Lowerboundtwo\left (\frac{1}{v}\right )$. Figure~\ref{fig: proof_ratio} (b) shows the lower bound~\eqref{eq:minimization_beta} on the ratio $\frac{\Ex{\Rev(\vals)}}{\opt}$ as a function of $\beta$. 
The minimum of $\approx 1/2.63$ is achieved for $\beta = \frac{1}{3}$. Thus we complete the proof for the identical marginals.
\begin{figure}[htbp]
\centering  
\subfigure[Quantile bounds]{   
\begin{minipage}{7cm}
\centering    
\includegraphics[scale=0.4]{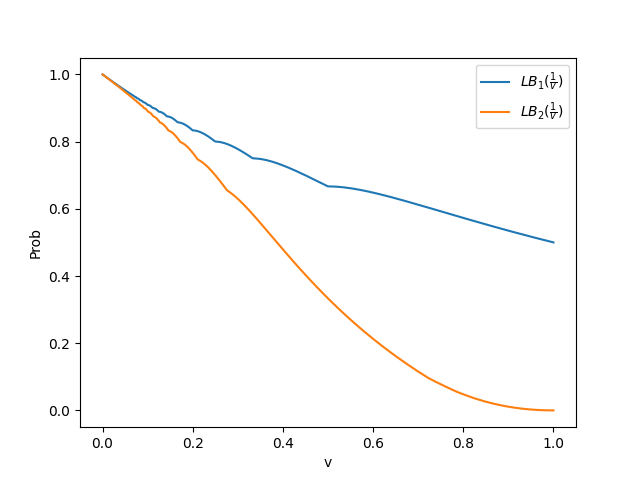}  
\end{minipage}
}
\subfigure[Lower bound on $\frac{\Ex{\Rev(\vals)}}{\opt}$]{ 
\begin{minipage}{7cm}
\centering    
\includegraphics[scale=0.4]{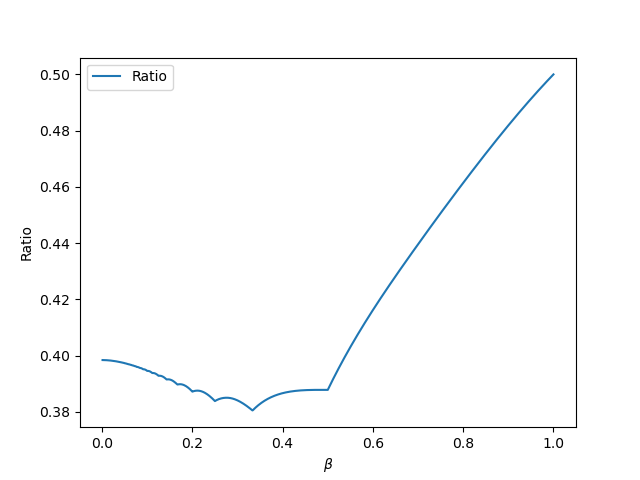}%
\end{minipage}
}
\caption{Curves in the proof of Theorem \ref{thm: myersoniid}}   
\label{fig: proof_ratio}
\end{figure}
\end{proof}


%% file: sec/ar.tex
We show in this section that the second price auction $\AR(\reserve)$ with any anonymous reserve $\reserve$ is pairwise-robust for any set of regular\footnote{The revenue gap between AR and Myerson's 
mechanism may be arbitrarily large for the mutually independent prior $\distind$ with irregular marginals} marginal distributions $(\disti)_{i\in[n]}$. 

\begin{theorem}
\label{thm: AR_main}
For any reserve price $\reserve\ge 0$ and regular marginals $(\disti)_{i\in[n]}$, 
$\AR(\reserve)$ is pairwise-robust with a constant $c\le 18.07$.
\end{theorem}
\begin{proof}
At a high level our proof proceeds as follows. 
\begin{enumerate}
\item We write the revenue of AR for $\dists$ and $\distind$ via quantile representations~\eqref{eq: pwi_rev}. 
Similar to the case of identical marginals,
we define the ex-ante relaxation threshold $\ww\eqdef \sup\{\val~|~\sum\limits_{i\in[n]}\Quani(\val)\geq 1\}$.
However, $\Qtwo(\tau)$ and $\Qtwoind(\tau)$ can be properly compared only when the expected number of bidders above the respective threshold $\tau$ is strictly larger than $1$ by the condition of Lemma~\ref{lem: ziyong}, i.e., for the thresholds $\tau<\ww$. 
\item We divide the revenue of the second price auction with anonymous reserve in~\eqref{eq: pwi_rev} for $\distind$ into two parts: the $\tailrev$ and the $\corerev$ 
\[
    \tailrev\eqdef \int_{\max\{\reserve, \ww\}}^{+\infty}\Qtwoind\left (\tau\right )\d \tau,\quad\quad\quad
    \corerev \eqdef \AR^{ind}\left (\reserve\right ) - \tailrev.
\]
We then argue that $\tailrev$ does not dominate the contribution to $\AR^{ind}(\reserve)$, i.e., the ratio $\tailrev/\corerev=O(1)$ is bounded by a constant.
After normalizing value distributions $\dists$ and $\distind$ so that $\max\{\ww, \reserve\} = 1$, we apply bounds from Lemma~\ref{lem: ziyong} for the 
expected number of bidders $\moww\eqdef\lim\limits_{\eps\to 0^+}\sum_{i\in[n]}\Quani(1+\eps)$ (usually, $\moww=\sum_{i\in[n]}\Quani(1)$ when none of $\disti$ have atoms).
\item We show that $\corerev$ can be covered by the  
revenue $\AR(\reserve)$ in the pairwise independent case.
\end{enumerate}

Throughout Section~\ref{sec: ar}, we assume without loss of generality that $\max\{\ww, \reserve\} = 1$, as one can simultaneously scale all values $\vali$ in $\vals\sim\dists$ and $\vals\sim\distind$ by the same multiplicative factor. We begin by establishing an upper bound on $\tailrev$ with respect to $\moww\eqdef\lim\limits_{\eps\to 0^+}\sum_{i\in[n]}\Quani(1+\eps)$.

\begin{claim}
\label{clm: T}
Assuming $\max\{\ww,\reserve\}=1$ and that all marginals are regular distributions
\begin{equation}
\label{eq: T}
\tailrev = \int_1^{ + \infty}\Qtwoind(\tau)\d \tau\leq  \frac{2\moww\left (1 - e^{-\moww}\cdot\left (1 + \moww\right )\right )}{\left (2 - \moww\right )^2} + \frac{\moww^2}{4}.
\end{equation}  

Furthermore, 
the maximum of the upper bound~\eqref{eq: T} is $\frac{9}{4} - \frac{4}{e}$ for $\moww=1$. 
\end{claim}
The complete proof of Claim~\ref{clm: T} is deferred to Appendix~\ref{app: corollary}.
\begin{proof}[Proof sketch.]
 We first prove that $ \Qtwoind\left (\tau\right )\leq\Qtwoind\left (1\right )\leq 1 - e^{-\moww}\left (1 + \moww\right )$ for any $\tau \ge 1$. Furthermore, in the range $\tau\ge 1 + \frac{2\moww}{\left (2 - \moww\right )^2}$ 
we have $\Qtwoind(\tau)\leq\left(
    \frac{2 - \moww}{\moww}\cdot \tau + 1\right)^{-2}.$
For the range $1\leq \tau\leq 1 + \frac{2\moww}{\left (2 - \moww\right )^2}$, we use the former bound; for the remaining $\tau \ge 1 + \frac{2\moww}{\left (2 - \moww\right )^2}$ we apply the latter. Hence,
\begin{multline*}
\tailrev=\int_{1}^{ 1 + \frac{2\moww}{\left (2 - \moww\right )^2}}\Qtwoind(\tau)\d \tau+
\int_{ 1 + \frac{2\moww}{\left (2 - \moww\right )^2}}^{+\infty}\Qtwoind(\tau)\d \tau \le
    \frac{2\moww}{\left (2 - \moww\right )^2}\cdot
    \left (1 - e^{-\moww}\cdot\left (1 + \moww\right )\right )+\\
    + \int_{1 + \frac{2\moww}{\left (2 - \moww\right )^2}}^{+\infty}\left(\frac{2 - \moww}{\moww}\cdot \tau + 1\right )^{-2}\d \tau= \frac{2\moww\left (1 - e^{-\moww}\cdot\left (1 + \moww\right )\right )}{\left (2 - \moww\right )^2} + \frac{\moww^2}{4},
\end{multline*}
which concludes the proof.
\end{proof}
Similar to the case of identical marginals, we consider two cases depending on the relation between $\reserve$ and $\ww$. 
We further differentiate the case when $\reserve< \ww$ into two cases:
(a) when all $\Quani(1)$ are small, and (b) when there is a bidder $i$ with large quantile $\Quani(1)$. 

\paragraph{Case 1: $\reserve > \ww$}
We first show that in this case, the $\corerev$ revenue 
is at least a constant of $\tailrev$ for the mutually independent prior $\distind$. 

\begin{claim}
\label{clm: r>w}
If $\reserve> \ww$, then
$
\frac{\tailrev}{\corerev}< \frac{9/4 - 4/e}{1 - 1/e} < 1.24 .$
\end{claim}
\begin{proof}
Since $\max\{\reserve,\ww\}=1$, we have $\reserve = 1$.
By definition of $\ww$ we have $\moww\le\sum_{i = 1}^n\Quani(1)<
\sum_{i = 1}^n\Quani(\ww)
\le 1$. 
We already have an upper bound~\eqref{eq: T} on $\tailrev$ that depends on $\moww$. There is also the following simple lower bound on $\corerev=\reserve\cdot \Qoneind\left (\reserve\right ) + \int_{\reserve}^{+\infty}\Qtwoind\left (\tau\right )\d \tau - \int_{\max\{\reserve,\ww\}}^{+\infty}\Qtwoind\left (\tau\right )\d \tau$.  
\begin{equation}
\label{eq: gap}
\corerev = 
\reserve\cdot \Qoneind\left (\reserve\right )=
1\cdot\left[1-\prod_{i\in[n]}(1-\Quani(1))\right]
\ge 1 - e^{-\moww},
\end{equation}
where the inequality holds, as after rearranging and taking the logarithm it is equivalent to $\sum_{i\in[n]}\ln(1-\Quani(1))\le -\moww=-\sum_{i\in[n]}\Quani(1)$ (note that $\ln(1-x)\le -x$ for $x\in[0,1]$).
These two bounds give an upper bound on $\tailrev/\corerev$:
\begin{equation*}
    \frac{\tailrev}{\corerev}\leq \frac{\frac{2\moww\left (1 - e^{-\moww}\cdot\left (1 + \moww\right )\right )}{\left (2 - \moww\right )^2} + \frac{\moww^2}{4}}{1 - e^{-\moww}}.
\end{equation*}
This upper bound achieves its maximum of $\frac{9/4 - 4/e}{1 - 1/e}$ for $\moww \rightarrow 1^{-}$.
\end{proof}
The pairwise independent revenue approximates $\corerev$, which directly follows from Lemma~\ref{lem: wine}
\begin{equation}
\label{eq: case1}
\AR(\reserve) = \reserve\cdot \Qone\left (\reserve\right ) + \int_{\reserve}^{+\infty}\Qtwo\left (\tau\right )\d \tau \geq \reserve\cdot \Qone\left (\reserve\right )\geq \frac{\reserve\cdot \Qoneind(\reserve)}{1.299} = \frac{\corerev}{1.299}.
\end{equation}
After combining this bound with the bound from Claim~\ref{clm: r>w} we get
\[
\frac{\AR^{ind}\left (\reserve\right )}{\AR\left (\reserve\right )} 
\underset{\text{\eqref{eq: case1}}}{\le} 
1.299\cdot \frac{\corerev+\tailrev}{\corerev} 
 \underset{\text{(Claim~\ref{clm: r>w})}}{\le}  1.299 \cdot \left (1 + 1.24\right )\le 2.91.
\]
This concludes the proof for the pairwise-robustness of AR in this case.

\paragraph{Case 2: $\reserve \leq \ww$}
Given our choice of scaling $\max\{\reserve,\ww\}=1$, we have $\ww = 1$ in this case. 
By Claim~\ref{clm: T} we have
\begin{equation*}
\AR^{ind}\left (\reserve\right ) = \reserve\cdot \Qoneind\left (\reserve\right ) + \int_\reserve^1\Qtwoind\left (\tau\right )\d \tau + \tailrev\leq 1 + \frac{9}{4} - \frac{4}{e} = \frac{13}{4} - \frac{4}{e},
\end{equation*}
where the inequality holds as $\reserve\cdot \Qoneind\left (\reserve\right ) + \int_\reserve^1\Qtwoind\left (\tau\right )\d \tau=
\int_0^\reserve \Qoneind\left (\reserve\right )\d \tau + \int_\reserve^1\Qtwoind\left (\tau\right )\d \tau\le\int_{0}^{1}1~\d \tau=1$.

Hence, it suffices to show that $\AR\left (\reserve\right )=\Omega(1)$. Since $\AR\left (\reserve\right ) = \reserve\cdot \Qone\left (\reserve\right ) + \int_{\reserve}^{+\infty}\Qtwo\left (\tau\right )\d \tau$, we can use Lemma~\ref{lem: ziyong} to lower bound $\Qtwo\left (\tau\right )$. However, the RHS of \eqref{eq: ziyong2} is related to $\momone(\tau)=\sum_{i\in[n]}\Quani(\tau)$ and requires $\momone(\tau)$ to be strictly larger than 1. Thus, we should first provide a lower bound for $\momone(\tau)$. We have the following lemma for regular marginals:
\begin{lemma}
\label{lem: reg_quantile}
For any regular distribution $\val\sim\dist$ with 
$p=\Quan(1)=\Prx[\val\sim\dist]{\val\ge 1}$ we have
\begin{equation}
\label{eq: quantile}
\forall \tau\ge 1\quad
\Quan(\tau) \le \frac{p}{\left (1 - p\right )\cdot \tau + p};
\quad\quad\quad
\forall \tau\le 1\quad
\Quan(\tau) \ge \frac{p}{\left (1 - p\right )\cdot \tau + p}
~
.
\end{equation}
\end{lemma}
\begin{proof} It is well known that the revenue as a function of the quantile $\Quan$ is a concave curve if and only if the respective distribution $\dist$ is regular. 
By the concavity (see Figure.~\ref{fig: proof} (a), the revenue-quantile curve must be above the line through points $(p, p)$ and $(1, 0)$ for quantiles $\Quan(\tau)$ with $\tau\ge 1$. Indeed, $\Rev(\Quan)\ge 0$ for the quantile $\Quan=1$ and $\Rev(p)=1\cdot\Quan(1)=p$ when the quantile is $\Quan(1)=p$.
For similar reasons, the revenue-quantile curve must be below the same line for quantiles $\Quan(\tau)$ with $\tau\le 1$.
I.e., 
\begin{equation*}
\forall \tau\ge 1\quad
    \tau\cdot \Quan(\tau)=\Rev(\Quan(\tau))\le -\frac{p}{1 - p}\cdot \Quan(\tau) + \frac{p}{1 - p},\quad\quad
\forall \tau\le 1\quad
    \tau\cdot \Quan(\tau)\ge
    -\frac{p}{1 - p}\cdot \Quan(\tau) + \frac{p}{1 - p},
\end{equation*}
which concludes the proof of Lemma~\ref{lem: reg_quantile}.

\begin{figure}[htbp]
\centering  
\subfigure[Regular revenue curve]{   
\begin{minipage}{7cm}
\centering    
\includegraphics[scale=0.5]{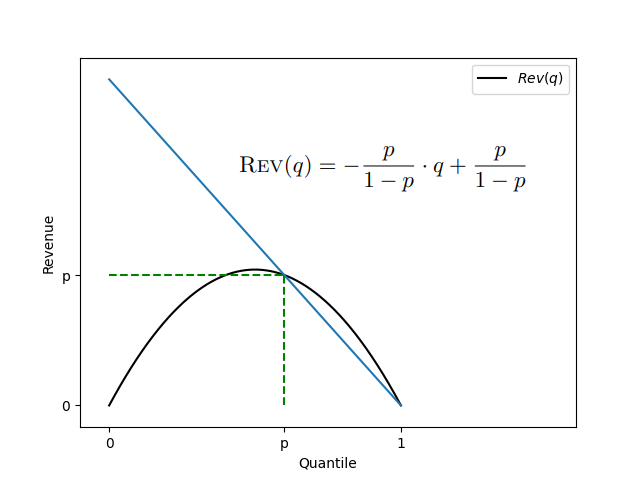}  
\end{minipage}
}
\subfigure[Convexity of $g(\cdot)$]{ 
\begin{minipage}{7cm}
\centering    
\includegraphics[scale=0.42]{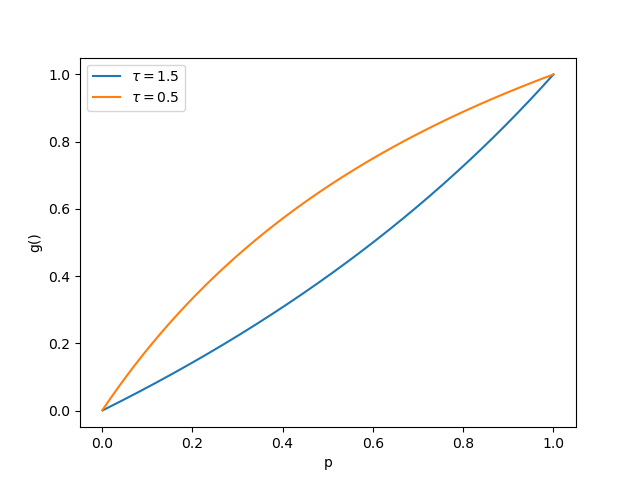}%
\end{minipage}
}
\caption{Curves in the proof of Lemma~\ref{lem: reg_quantile}}   
\label{fig: proof}
\end{figure}
\end{proof}
In order to provide a lower bound on $\momone(\tau)$, we check convexity of the function $g(p)\eqdef\frac{p}{(1-p)\tau+p}$.
\begin{claim}
\label{pro: convex}
Function $g(p)=\frac{p}{(1-p)\tau+p}$ is concave when $\tau\le 1$ and is convex when $\tau\ge 1$ on the interval $p\in[0,1]$.
\end{claim}
\begin{proof}
The second derivative $\frac{\d^2 }{\d p^2}g(p)=\frac{2\tau(\tau - 1)}{(\tau - (\tau - 1)p)^3}$ is non-negative for $\tau\ge 1$ and is non-positive for  $\tau\le 1$, when $p\in[0,1]$. 
Figure~\ref{fig: proof} (b) shows two examples of $g(p)$ with different convexity.
\end{proof}

Now we are ready to give a lower bound on $\momone(\tau)=\sum_{i\in[n]} \Quani(\tau)$ for $\tau\in[\reserve,\ww]$. Let $\probi \eqdef\lim\limits_{\eps\to 0^+} \Quani(1+\eps)$ to simplify the notations. Then $\sum_{i\in[n]}\probi\le 1$. By Lemma~\ref{lem: reg_quantile} for $\tau\le \ww=1$, we have
\begin{equation}
\label{eq: sum_q}
\sum_{i = 1}^n\Quani\left (\tau\right ) \geq \sum_{i =1}^n\frac{\probi}{\left (1 - \probi\right )\cdot \tau + \probi}.
\end{equation}
The function $\frac{\probi}{\left (1 - \probi\right )\tau + \probi}$ in \eqref{eq: sum_q} is concave in $\probi$ when $\tau \le 1$ by Claim~\ref{pro: convex}. 
Thus the RHS of~\eqref{eq: sum_q} is minimized under $\sum_{i\in[n]}\probi\le 1$ when $\probi[1] = 1$ and $\probi = 0$ for $i > 1$. I.e., in the worst case, \eqref{eq: sum_q} gives us $\sum_{i = 1}^n\Quani\left (\tau\right )= 1$, which 
does not allow us to use Lemma~\ref{lem: ziyong}.
In order to make the lower bound on $\momone(\tau)=\sum_{i\in[n]}\Quani(\tau)$ strictly larger than 1, we assume that $\max_i\{\probi\} < \pp < 1$ where $\pp$ is a parameter to be determined later. Next, we consider the case when $\max_i\{\probi\} < \pp$.

\paragraph{Case 2.(a): $\max_i\{\Quani(\ww)\} < \pp$.}
By Claim~\ref{pro: convex}, $\frac{p}{(1 - p)\tau + p}$ is concave in $p$ when $p, \tau\le 1$. Then
\[
\sum_{i = 1}^n\Quani\left (\tau\right ) 
\underset{\text{Lemma~\ref{lem: reg_quantile}}}{\ge} 
\sum_{i =1}^n\frac{\probi}{\left (1 - \probi\right )\cdot \tau + \probi}
\underset{\text{Claim~\ref{pro: convex}}}{\ge} 
\frac{\pp}{\left (1 - \pp\right )\cdot \tau + \pp}+\frac{1 - \pp}{\pp\cdot\tau + \left (1 - \pp\right )}.
\]

Now we can get a lower bound on $\AR\left (\reserve\right )$ as follows.
\begin{multline}
\label{eq: ar_case_a}
\AR\left (\reserve\right ) = \reserve\cdot \Qone\left (\reserve\right ) + \int_{\reserve}^{+\infty}\Qtwo\left (\tau\right )\d \tau
\geq \reserve\cdot \Qone\left (\reserve\right ) + \int_{\reserve}^{1}\Qtwo\left (\tau\right )\d \tau \\
\geq \reserve\cdot \Qone\left (\reserve\right ) + \int_{\reserve}^1\Lowerboundtwo\left (\frac{\pp}{\left (1 - \pp\right )\cdot \tau + \pp}+\frac{1 - \pp}{\pp\cdot \tau + \left (1 - \pp\right )}\right )\d \tau.
\end{multline}
We would like to show that $
\AR\left (\reserve\right )\geq \int_{0}^1\Lowerboundtwo\left (\frac{\pp}{\left (1 - \pp\right )\cdot \tau + \pp}+\frac{1 - \pp}{\pp\cdot \tau + \left (1 - \pp\right )}\right )\d \tau. 
$ Indeed, in the bound~\eqref{eq: ar_case_a} the first term
$\reserve\cdot \Qone\left (\reserve\right )=\int_0^\reserve \Qone\left (\reserve\right ) \d \tau$. Furthermore, $\frac{\pp}{\left (1 - \pp\right )\cdot \tau + \pp}+\frac{1 - \pp}{\pp\cdot \tau + \left (1 - \pp\right )}\le 2$ for any $\tau\ge 0$ and $\pp\in[0,1]$. Hence,
\begin{equation*}
\Qone\left (\reserve\right )\geq \Qone\left (1\right )\geq \Lowerboundone\left (1\right )  = \frac{1}{2} > \frac{1}{3} = \Lowerboundtwo\left (2\right ) \geq \Lowerboundtwo\left (\frac{\pp}{\left (1 - \pp\right )\cdot \tau + \pp}+\frac{1 - \pp}{\pp\cdot \tau + \left (1 - \pp\right )}\right ), 
\end{equation*}
where the first inequality holds as $\Qone(\tau)$ is monotonically decreasing in $\tau$ and $\reserve\le 1$; the second inequality is from Lemma~\ref{lem: ziyong};
the last inequality holds as $\Lowerboundtwo(\moww)$ is monotonically increasing in $\moww$. 

We numerically calculate the integral $\int_{0}^1\Lowerboundtwo\left (\frac{\pp}{\left (1 - \pp\right )\cdot \tau + \pp}+\frac{1 - \pp}{\pp\cdot \tau + \left (1 - \pp\right )}\right )\d \tau$  for a fixed $\pp = 0.674$, and get that it is larger than $0.0984$.
Together with an upper bound on $\AR^{ind}\left (\reserve\right )\le \frac{13}{4} - \frac{4}{e}$ we get the ratio 
$\frac{\AR^{ind}\left (\reserve\right )}{\AR\left (\reserve\right )}\le 18.07$.

\paragraph{Case 2.(b): $\max_i\{\Quani(\ww)\} \ge \pp$.}
Similar to case 1, we first show in Claim~\ref{clm: max>p} that $\frac{\tailrev}{\corerev}=\Omega(1)$ for the mutually independent prior $\distind$. 
The proof of this claim is deferred to Appendix~\ref{app: case3}.
\begin{claim}
\label{clm: max>p}
If $\max_i\{\Quani\left (\ww\right )\} \geq \pp$ and $\reserve\le\ww=1$,
for regular marginals
$\frac{\tailrev}{\corerev} \leq \frac{2\pp^2 - 2\pp + 1}{\pp^3}\cdot \frac{e}{e - 1}.
$
\end{claim}
Next we provide a lower bound on the ratio $\frac{\Qtwo\left (\tau\right )}{\Qtwoind(\tau)}$ for each fixed threshold $\tau\in[\reserve,\ww]$ similar to Lemma~\ref{lem: wine}, but with dependency on the constant parameter $\pp$. 
\begin{lemma}
\label{lem: q2lb2}
Suppose $\max_i\{\Quani\left (\tau\right )\} \geq \pp$. Then for any pairwise prior $\dists$ with regular marginals
\begin{equation*}
   \frac{\Qtwo\left (\tau\right )}{\Qtwoind\left (\tau\right )}\geq \min\left\{\frac{1}{1.299} - \frac{\left (1 - \pp\right )\cdot e}{e - 1}~~,~ \Lowerboundtwo\left (1 + \pp\right )\right\}\eqdef\Lowerboundthree\left (\pp\right ).
\end{equation*}
\end{lemma}
\begin{proof}
Without loss of generality, let us assume that the first bidder has the largest quantile $\Quani(\tau)$. We would like to give a lower bound on $\Qtwo\left (\tau\right )$. To this end, let us consider the following probabilities. 
$$
\hQone\left (\tau\right ) \eqdef \Prx[\vals\sim\dists]{|\{i>1: \vali\geq \tau\}|\geq 1},\quad\quad
\widehat{Q}_1^{ind}\left (\tau\right )\eqdef \Prx[\vals\sim\distind]{|\{i>1: \vali\geq \tau\}|\geq 1}.
$$
Clearly, $\Qtwoind\left (\tau\right ) \leq \widehat{Q}_1^{ind}\left (\tau\right )$ and $\Qtwo\left (\tau\right ) \leq \widehat{Q}_1\left (\tau\right )$, since if at least two bidders have values larger than $\tau$, then at least one of them is not the first bidder.
Similarly, we observe that
\begin{equation*}
 \Qtwo\left (\tau\right ) \ge \hQone\left (\tau\right ) - \Prx[\vals\sim\dists]{|\{i>1: \vali\geq \tau\}|\geq 1, \vali[1] < \tau}.
\end{equation*}
Indeed, the event that at least two bidders exceed the threshold $\tau$ includes the event that at least one bidder with $i>1$ exceeds $\tau$ and $\vali[1]\ge\tau$.
Let us denote as $\hmoww\eqdef \sum_{i = 2}^n\Quani\left (\tau\right )$. We first give a lower bound on $\Qtwo(\tau)$ for the case $\hmoww \leq 1$.
In this case,
\[
\prob{|\{i>1: \vali\geq \tau\}|\geq 1, \vali[1] < \tau}
\leq \sum_{i = 2}^n\prob{\val_i\geq\tau, \vali[1] < \tau}
 = \sum_{i = 2}^n \Quani\left (\tau\right )\cdot (1 - \pp)= (1 - \pp)\cdot\hmoww,
\]
where the first inequality is due to the union bound; the second equality holds because of the pairwise independence of $\dists$; the last equality is the definition of $\hmoww$. Thus,
\begin{equation*}
 \Qtwo\left (\tau\right ) \ge \hQone\left (\tau\right ) - \prob{|\{i>1: \vali\geq \tau\}|\geq 1, \vali[1] < \tau}\ge
 \hQone\left (\tau\right )-(1 - \pp)\cdot\hmoww.
\end{equation*}
Combining this lower bound on $\Qtwo(\tau)$ with the upper bound 
on $\Qtwoind\left (\tau\right )\le\hQone(\tau)$ we get the following
\[
\frac{\Qtwo\left (\tau\right )}{\Qtwoind\left (\tau\right )}\geq \frac{\hQone\left (\tau\right ) - \left (1 - \pp\right )\hmoww}{\widehat{Q}_1^{ind}\left (\tau\right )}
\ge
\frac{\frac{\widehat{Q}_1^{ind}\left (\tau\right )}{1.299} - \frac{\left (1 - \pp\right )\cdot e}{e - 1}\cdot \widehat{Q}_1^{ind}\left (\tau\right )}{\widehat{Q}_1^{ind}\left (\tau\right )}
=\frac{1}{1.299} - \frac{\left (1 - \pp\right )\cdot e}{e - 1},
\]
where the second inequality holds by Lemma~\ref{lem: wine} applied to  $\hQone\left (\tau\right )$ and a similar to~\eqref{eq: gap} argument that $\widehat{Q}_1^{ind}(\tau)=1-\prod_{i>1}(1-\Quani(\tau))\ge 1-e^{-\hmoww}\ge\left(1-\frac{1}{e}\right)\hmoww$\quad (note that function $\frac{1-e^{-x}}{x}$ attains its minimum of $1-1/e$ at $x=1$ on the interval $x\in[0,1]$).
If $\hmoww > 1$, according to Lemma~\ref{lem: ziyong} we directly have:
$
\frac{\Qtwo\left (\val\right )}{\Qtwoind\left (\val\right )}\geq \Lowerboundtwo\left (1 + \pp\right ),
$
which completes the proof. 
\end{proof}
Specially, if we pick $\pp = 0.674$ in Lemma~\ref{lem: q2lb2}, then $\Lowerboundthree(\pp)\ge 0.215$. Now we are ready to complete the proof of Theorem~\ref{thm: AR_main}. 
\begin{align*}
\frac{\AR^{ind}\left (\reserve\right )}{\AR\left (\reserve\right )} &= \frac{\corerev+\tailrev}{\reserve\cdot \Qone\left (\reserve\right ) + \int_{\reserve}^{+\infty}\Qtwo\left (\tau\right )\d \tau}  \tag{equation~\eqref{eq: pwi_rev}}\\
&\leq \frac{\corerev+\tailrev}{\frac{1}{1.299}\cdot r\cdot \Qoneind\left (\reserve\right )  + \Lowerboundthree\left (\pp\right )\cdot \int_\reserve^\ww\Qtwoind\left (\tau\right )\d \tau} \tag{lemma \ref{lem: wine} and lemma \ref{lem: q2lb2}}\\
&\leq \frac{1}{\Lowerboundthree\left (\pp\right )}\cdot \left (1 + \frac{\tailrev}{\corerev}\right )\tag{$\Lowerboundthree\left (\pp\right )\leq \frac{1}{1.299}$}\\
& \leq \frac{1}{\Lowerboundthree\left (\pp\right )} \cdot \left (1 + \frac{2\pp^2 - 2\pp + 1}{\pp^3}\cdot \frac{e}{e - 1}\right )\tag{Claim~\ref{clm: max>p}}
\end{align*}
The ratio is smaller than $18.07$ for $\pp = 0.674$ which completes the proof of Theorem~\ref{thm: AR_main}.
\end{proof}

%% file: sec/appendixnew.tex
As we shown in Section~\ref{sec: ar}, Claim~\ref{clm: T} follows directly from the following 2 claims:

\begin{claim}
\label{clm: range1}
$\forall\tau\geq 1,$ we have $ \Qtwoind\left (\tau\right )\leq\Qtwoind\left (1\right )\leq 1 - e^{-\moww}\left (1 + \moww\right ).$
\end{claim}
The second inequality takes equality when $n\rightarrow+\infty$ and $q_i(1) = \frac{\moww}{n}$ for each $i\in[n]$. Then $\Qtwoind(1) = 1 - (1 - \frac{\moww}{n})^n - n\cdot\frac{\moww}{n}\cdot(1 - \frac{\moww}{n})^{n - 1}\rightarrow 1 - e^{-\moww}\left (1 + \moww\right )$. This bound is quite loose, as we apply the same bound for all $\tau\ge 1$. 
The following is a tighter bound for sufficiently large $\tau$:
\begin{claim}
\label{clm: range2}
In the range $\tau\ge 1 + \frac{2\moww}{\left (2 - \moww\right )^2}$ 
we have $\Qtwoind(\tau)\leq\left(
    \frac{2 - \moww}{\moww}\cdot \tau + 1\right)^{-2}.$
\end{claim}
The bound in Claim~\ref{clm: range2}
achieves equality when $p_1 = p_2 = \frac{s}{2}$ and $\Quan_1 = \Quan_2 = \left(
    \frac{2 - \moww}{\moww}\cdot \val + 1\right)^{-1}$.

Note that both Claim~\ref{clm: range1} and Claim~\ref{clm: range2} are upper bound for $\Qtwoind(\tau)$.
The following lemma is a closed-form expression for $\Qtwoind(\tau)$ given in~\cite{jin2020tight}. 
\begin{lemma}[\cite{jin2020tight}]
\label{lem: ind_prob}
The probability that at least two bidders have values at least $\tau$  under $\distind$ is
\begin{equation}
\label{eq: jin}
\Qtwoind(\tau) = 1 - \left (\prod_{i = 1}^n\left (1 - \Quani \left (\tau\right )\right )\right )\cdot\left (1 + \sum_{i = 1}^n\left (\frac{\Quani \left (\tau\right )}{1 - \Quani \left (\tau\right )}\right )\right ).
\end{equation}
\end{lemma} 
Based on this closed-form expression, we provide the proof of Claim~\ref{clm: range1} and Claim~\ref{clm: range2} respectively.

\subsection{Proof of Claim~\ref{clm: range1}}
\label{app: upper}
\input{sec/appendix}

\subsection{Proof of Claim~\ref{clm: range2}}
\label{app: upper_q2ind}
\input{sec/appendix2}

%% file: sec/appendix.tex
We give proof to a slightly modified Claim~\ref{clm: range1} in this appendix. 
The modification is that for $\tau\ge 1=\max\{\ww,\reserve\}$ we must have $\sumtau=\sum_{i = 1}^n \Quani\left (\tau\right)\le 1$.
The modification is a general version of Claim~\ref{clm: range1} which holds for all $\tau\ge 1=\max\{\ww,\reserve\}$ and $\sumtau=\sum_{i = 1}^n \Quani\left (\tau\right)\le 1$.
\begin{claim}
\label{lem: upper_bound_two_events}
Suppose $\sumtau = \sum_{i = 1}^n \Quani\left (\tau\right )\le 1$. 
Then $\Qtwoind(\tau)\le 1 - e^{-\sumtau}(1 + \sumtau)$.    
\end{claim}
\begin{proof}
Denote $\vq = \left (\Quani[1]\left (\tau\right ),\ldots,\Quani[n]\left (\tau\right)\right)$. Then by Lemma~\ref{lem: ind_prob} $\Qtwoind\left(\vq\right)=\Qtwoind\left(\tau\right)$ is as follows.  
\begin{equation*}
    \Qtwoind\left (\vq\right ) = 1 - \left (\prod_{i = 1}^n\left (1 - \Quani\left (\tau\right )\right )\right )\cdot\left (1 + \sum_{i = 1}^n\left (\frac{\Quani\left (\tau\right )}{1 - \Quani\left (\tau\right )}\right )\right ).
\end{equation*}
It turns out that $\Qtwoind(\vq)$ takes maximum when some elements in $\vq$ are zeros and the rest are all equal.
\begin{claim}
\label{clm: zero_or_equal}
$\Qtwo\left (\vq\right )$ is maximized when exactly $m$ of $\Quani=\sumtau / \mm $ and the rest are zeros. I.e.,
\begin{equation}
\label{eq: \Qtwo_m}
\Qtwoind\left (\tau \right ) \leq \max_{m\le n}\left\{ 1 - \left (1 - \frac{\sumtau}{\mm }\right )^\mm \left (1 + \frac{\mm\cdot \sumtau}{\mm  - \sumtau}\right )\right \}.
\end{equation}
\end{claim}
\begin{proof}

We apply the idea of successive replacement of elements. If the vector $\vq$ does not satisfy the condition in the claim, then we consider three other vectors $\vq$ and show that $\Qtwoind(\vq)$ increases for at least one of them.
Let $\vq_1 = \left (\frac{q_1\left (\tau\right ) + q_2\left (\tau\right )}{2}, \frac{q_1\left (\tau\right ) + q_2\left (\tau\right )}{2}, q_3\left (\tau\right ),\ldots\right )$. 
If $q_1\left (\tau\right ) + q_2\left (\tau\right ) < 1$, then let $\vq_2 = \left (q_1\left (\tau\right ) + q_2\left (\tau\right ), 0, q_3\left (\tau\right ),\ldots\right )$. Otherwise, let $\vq_3 = \left (q_1\left (\tau\right ) + q_2\left (\tau\right )  - 1, 1, q_3\left (\tau\right ),\ldots\right )$.
Then
\begin{align*}
\Qtwoind\left (\vq_1\right ) - \Qtwoind\left (\vq\right ) &= 
\left (\prod_{i \ge 3}\left (1 - \Quani\left (\tau\right )\right )\right)\cdot
\left (1 - \sum_{i = 3}^n\left (\frac{q_i}{1 - q_i}\right )\right )\cdot\left (\left (\frac{q_1 + q_2}{2}\right )^2 - q_1q_2\right ),\\
\Qtwoind\left (\vq_2\right ) - \Qtwoind\left (\vq\right ) &= 
\left (\prod_{i \ge 3}\left (1 - \Quani\left (\tau\right )\right )\right)\cdot
\left (\sum_{i = 3}^n\left (\frac{q_i}{1 - q_i}\right ) - 1\right )\cdot q_1q_2,\\
\Qtwoind\left (\vq_3\right ) - \Qtwoind\left (\vq\right ) &= 
\left (\prod_{i \ge 3}\left (1 - \Quani\left (\tau\right )\right )\right)\cdot
\left (\sum_{i = 3}^n\left (\frac{q_i}{1 - q_i}\right ) - 1\right )\cdot\left (1 - q_1\right )\left (1 - q_2\right ).
\end{align*}
Now observe that if $\Qtwoind\left (\vq_1\right ) - \Qtwoind\left (\vq\right ) \le 0$, then $\left (1 - \sum_{i = 3}^n\left (\frac{q_i}{1 - q_i}\right )\right ) \ge 0$. But then either $\Qtwoind\left (\vq_2\right ) - \Qtwoind\left (\vq\right ) \ge 0$, or $\Qtwoind\left (\vq_3\right ) - \Qtwoind\left (\vq\right ) \ge 0$. 
Hence, either $\Qtwoind(\vp_1)>\Qtwoind(\vp)$, or one of 
$\Qtwoind(\vp_2)$ and $\Qtwoind(\vp_3)$ is 
is larger than or equal to the original $\Qtwoind(\vp)$. I.e., if a vector $\vp\in[0,1]^n$ with $\sum_{i\in[n]}\Quani=\sumtau$ maximizes $\Qtwoind(\vp)$ and has at least two non-equal positive coordinates, then we should be able to find another such $\vp$ that still maximizes $\Qtwoind(\vp)$ and either has more $0$ coordinates than original $\vp$ or has the same number of $0$ coordinates and more $1$ coordinates. After a finite number of steps, it should not be possible to make any such modifications, which concludes the proof of Claim~\ref{clm: zero_or_equal}. 
\end{proof}

Now we return to the proof of Claim~\ref{lem: upper_bound_two_events}.
By Claim~\ref{clm: zero_or_equal} it suffices to optimize $\mm$. 
We regard $\mm\ge 1$ as a real number and denote the function from~\eqref{eq: \Qtwo_m} as
$f(\mm) \eqdef 1 - \left (1 - \frac{\sumtau}{\mm }\right )^\mm \left (1 + \frac{\mm \sumtau}{\mm  - \sumtau}\right ).$    
Maximization of $f(\mm)$ is equivalent to minimization of $g(\mm)=\ln(1-f(\mm))=\mm\cdot\ln\left(1-\frac{\sumtau}{m}\right)+\ln\left(1+\frac{m\cdot\sumtau}{m-\sumtau}\right)$. We can verify that derivative of $g'(m)\le 0$ for any $\mm\geq 1$. Indeed,
\begin{align*}
\frac{\d g\left (\mm\right )}{\d \mm} &= \ln\left(1-\frac{\sumtau}{\mm}\right) - \frac{\sumtau(\sumtau\cdot\mm - 2\sumtau+\mm)}{(\sumtau - \mm)(\sumtau\cdot\mm - \sumtau +\mm)}  \\
&\leq \frac{-\frac{\sumtau}{\mm}\left(6-\frac{\sumtau}{\mm}\right)}{6-\frac{4\sumtau}{\mm}}- \frac{\sumtau(\sumtau\cdot\mm - 2\sumtau+\mm)}{(\sumtau - \mm)(\sumtau\cdot\mm - \sumtau +\mm)}\\
&=-\frac{\sumtau^2(3\mm^2(1 - \sumtau) + (\mm - 1)\sumtau^2)}{2\mm(3\mm - 2\sumtau)(\mm - \sumtau)(\sumtau(\mm - 1) + \mm)}\\
&\leq 0,
\end{align*}
where the first inequality is known upper estimate of logarithm $\ln(1 + x)\leq \frac{x(6 + x)}{6 + 4x}$ for $x > -1$; the second inequality is easy to verify for $\sumtau\leq 1, \mm\geq 1$. As a result, $\Qtwoind(\tau)$ achieves its maximum for  $\sumtau \leq 1$, when $\mm\rightarrow +\infty$. I.e., $f(\mm)\le 1-e^{-\sumtau}(1+\sumtau)=\lim_{\mm\to\infty} f(\mm)$,  which completes the proof of Lemma~\ref{lem: upper_bound_two_events}.
\end{proof}

%% file: sec/appendix2.tex
We first use bounds~\eqref{eq: quantile} in the expression~\eqref{eq: jin} for $\Qtwoind(\tau)$ and regard each $\probi=\Quani(1)$ as optimization variable to maximize $\Qtwoind(\tau)$.
We apply the method of successive replacement of elements to prove Claim~\ref{clm: range2}. I.e., by maintaining the $\momone=\sum_{i\in[n]}\probi$ we modify a few of $(p_i)$ at a time so that the expression \eqref{eq: jin} only increases; we eventually narrow down the set of $(p_i)_{i\in[n]}$ that could maximize $\Qtwoind(\tau)$ and verify that the desired bound holds for any vector of probabilities $(p_i)_{i\in[n]}$ in this set.
Specifically, we prove that if $\tau\ge 1 + \frac{2\moww}{\left (2 - \moww\right )^2}$, then we can keep merging variables until at most two non-zero $\probi$ remain.
\begin{proof}
We start by applying upper bounds~\eqref{eq: quantile} on the quantiles $\Quani(\tau)\le\ubqi(\tau)\eqdef\frac{p_i}{\left (1 - p_i\right )\cdot \tau + p_i}$ for $\tau>1$, where $p_i=\Quani(1)$ (more accurately we have $\probi = \lim\limits_{x\rightarrow 1^+}\Quani(x)$ as we only care about $\tailrev$ in this proof). 
Then $\Qtwoind(\tau)=\Qtwoind(\vq)\le\Qtwoind(\vubq)$, where
$\vubq\eqdef(\ubqi(\tau))_{i\in[n]}$ coordinate-wise dominates vector of probabilities $\vq$, by definition of $\Qtwoind(\vq)$.

We also have a closed form~\eqref{eq: jin} for the $\Qtwoind(\vubq)$, which after simple algebraic transformation can be written as follows
\begin{equation*}
   \forall \tau \ge 1,\quad \Qtwoind(\tau)=\Qtwoind(\vq)\le
   \Qtwoind(\vubq)
   = 1 - \frac{1 + \sum\limits_{i = 1}^n\frac{\probi}{\tau\left (1 - \probi\right )}}{\prod\limits_{i = 1}^n\left (1 + \frac{\probi}{\tau\left (1 - \probi\right )}\right )}.
\end{equation*}

We consider this upper bound on $\Qtwoind(\tau)$ as the following optimization problem with respect to $\vp = (p_1, \ldots ,p_n)$, where $\moww = \sum_{i = 1}^n\probi\le 1$. We shall prove that its value is not more than our desired upper bound.
\begin{align}
\label{eq: optimization_C}
&\max_{\vp}\quad 1 - \frac{1 + \sum\limits_{i = 1}^n\frac{\probi}{\tau\left (1 - \probi\right )}}{\prod\limits_{i = 1}^n\left (1 + \frac{\probi}{\tau\left (1 - \probi\right )}\right )}\\
& \begin{array}{r@{\quad}r@{}l@{\quad}l}
s.t.&\sum\limits_{i = 1}^n \probi&= \moww,&\notag\\
 &\probi&\geq 0.&i = 1, \ldots, n\notag
\end{array}
\end{align}

To show this, we apply the method of successive replacement of elements. First, we try to merge two $\probi$, i.e., changing $\vp = \left (\probone, \probtwo,\ldots p_n\right )$ to $\left (\probone + \probtwo, 0,\ldots p_n\right )$ as much as we can while (weakly) increasing the objective. 
Specifically, the increase of this objective after such merging transformation is equivalent to the following inequality
\begin{equation}
\label{eq: sroe}
   \frac{1 + \alfa  + \frac{\probone}{\tau\left (1 - \probone\right )} + \frac{\probtwo}{\tau\left (1 - \probtwo\right )}}{\left (1 + \frac{\probone}{\tau\left (1 - \probone\right )}\right )\left (1 + \frac{\probtwo}{\tau\left (1 - \probtwo\right )}\right )}
 \ge 
\frac{1 + \alfa  + \frac{\probone + \probtwo}{\tau\left (1 - \probone - \probtwo\right )}}{1 + \frac{\probone + \probtwo}{\tau\left (1 - \probone - \probtwo\right )}},
\end{equation}
where $\alfa \eqdef \sum_{i = 3}^n\frac{\probi}{\tau\left (1 - \probi\right )}$ is the respective term for vector $\vp$ not affected by the transformation. Our goal will be to minimize the number $\mm$ of strictly positive elements in the vector $\vp$.
\begin{claim}
Let $\ttt \eqdef \probone + \probtwo\leq \moww$, then
the sufficient and necessary condition for inequality ($\ref{eq: sroe}$) is that:
\begin{equation}
\label{eq: sn_condition}
    \tau\cdot\Big(\left (\alfa\cdot \tau - 1\right )\left (1 - \ttt\right ) + \alfa\cdot \tau\Big)\geq \ttt + \alfa\cdot \tau\left (1 - \ttt\right ).
\end{equation}
\end{claim}
\begin{proof}
After multiplying both sides in \eqref{eq: sroe} by their denominators we get
\begin{align*}
    \eqref{eq: sroe}&\Leftrightarrow \left (1 + \alfa  + \frac{\probone + \probtwo}{\tau\left (1 - \probone - \probtwo\right )}\right )\left (1 + \frac{\probone}{\tau\left (1 - \probone\right )}\right )\left (1 + \frac{\probtwo}{\tau\left (1 - \probtwo\right )}\right )\\
    &\quad \leq \left ( 1 + \alfa  + \frac{\probone}{\tau\left (1 - \probone\right )} + \frac{\probtwo}{\tau\left (1 - \probtwo\right )}\right )\left (1 + \frac{\probone + \probtwo}{\tau\left (1 - \probone - \probtwo\right )} \right ) \\
    &\Leftrightarrow \frac{\probone\probtwo\left(\alfa\tau^2\left (\ttt - 2 \right ) - \left (\alfa + 1 \right )\left ( \ttt - 2\right )\tau + \ttt\right )}{\left (1 - \probone \right )\left (1-\probtwo \right )\left ( 1-\ttt\right )\tau}\leq 0\\
    &\Leftrightarrow \alfa\tau^2\left (\ttt - 2 \right ) - \left (\alfa + 1 \right )\left ( \ttt - 2\right )\tau + \ttt\leq 0\\
    &\Leftrightarrow \tau\left (\left (\alfa \tau - 1\right )\left (1 - \ttt\right ) + \alfa \tau\right )\geq \ttt + \alfa \tau\left (1 - \ttt\right ).
\end{align*}
\end{proof}
We have the following upper and lower bounds on the product $\alfa\cdot\tau$ from the inequality~\eqref{eq: sn_condition}.
\begin{claim} [Bounds on $\alfa\cdot \tau$]
\label{clm: sv_bound}
Assume that $\vp$ has exactly $\mm$ strictly positive elements. Then
\begin{equation*}
1 - \ttt\leq \frac{\left (\mm - 2\right )\left (1 - \ttt\right )}{\left (\mm - 2\right ) - \left (1 - \ttt\right )}\leq \alfa\cdot \tau = \sum_{i = 3}^\mm\frac{\probi}{1 - \probi}\leq \frac{1 - \ttt}{\ttt}.
\end{equation*}
\end{claim}
\begin{proof}
The function $f\left (x\right ) = \frac{x}{1 - x}$ is a convex on $[0, 1)$. Thus the sum of $\frac{p_i}{1-p_i}$ achieves maximum when all $p_i=0$ for $i>3$ and $p_3 = 1 - \ttt$. This gives us the claimed upper bound on $\alfa\cdot \tau$. 
The sum achieves its minimum, when all $\probi = \frac{1 - \ttt}{\mm - 2}$, letting $\mm$ converge to $+\infty$ leads to the desired lower bound.
\end{proof}

The following claim gives a sufficient condition for the transformation to increase the objective.
\begin{claim}
\label{clm: m_gen}
For 
any
$\tau\geq 1 + \frac{\ttt}{\left (1 - \ttt\right )^2}$ the inequality \eqref{eq: sn_condition} holds true.
\end{claim}
\begin{proof}
According to Claim~\ref{clm: sv_bound}, we have $\alfa\cdot \tau\geq 1 - \ttt$. Thus
\begin{equation*}
 \text{LHS of \eqref{eq: sn_condition}} =
 \tau\Big(\alfa\cdot \tau\left (1 - \ttt\right ) - \left (1 - \ttt\right ) + \alfa\cdot \tau\Big) \geq \tau^2\alfa \left (1 - \ttt\right ).   
\end{equation*}
By the assumption of the Claim~\ref{clm: m_gen}, we have $\left(1 - \ttt\right )^2\left (\tau - 1\right )\geq \ttt$. Since $\alfa\cdot\tau\ge 1-t$, we have $\alfa\cdot\tau\cdot(1-\ttt)(\tau-1)\ge(1-t)^2(\tau-1)\ge t$. By rearranging the terms in inequality $\alfa\cdot\tau\cdot(1-\ttt)(\tau-1)\ge t$, we get $\tau^2\alfa (1 - \ttt)\geq \ttt + \alfa \tau(1 - \ttt)$. Hence, we get~\eqref{eq: sn_condition} as its LHS is at least $\tau^2\alfa (1 - \ttt)\geq \ttt + \alfa \tau(1 - \ttt)=\text{RHS of \eqref{eq: sn_condition}}$.
\end{proof}

If the number of non-zero elements $\mm > 3$, then for at least one pair of them (say $p_1$ and $p_2$), their sum $p_1+p_2=\ttt\leq \frac{\momone}{2}$. Then according to Claim~\ref{clm: m_gen}, we can merge them and (weakly) increase our objective. 
Hence, we can assume that $\mm\le 3$.
 When $\mm=3$ we can get an improvement to the sufficient condition of Claim~\ref{clm: m_gen}. Namely, we have   
\begin{claim}
\label{clm: m_3}
If there are exactly three non-zero elements in $\vp$ (i.e., $\mm = 3$), then $\tau\geq \frac{\left (1 - \ttt\right )^2 + \ttt^2}{2\left (1 - \ttt\right )^2}$ is a sufficient and necessary condition for inequality (\ref{eq: sn_condition}).
\end{claim}
\begin{proof}
If $\mm = 3$, then $\alfa\cdot \tau = \frac{1 - \ttt}{\ttt}$ and inequality \eqref{eq: sn_condition} is equivalent to our assumption $\tau\geq \frac{\left (1 - \ttt\right )^2 + \ttt^2}{2\left (1 - \ttt\right )^2}$.
\end{proof}

Now, if $\mm = 3$, there must exist two non-zero elements (say $p_1$ and $p_2$) whose sum $\ttt\leq \frac{2\moww}{3}$. It is easy to verify that they will satisfy conditions of Claim~\ref{clm: m_3}, since function $\frac{\left (1 - \ttt\right )^2 + \ttt^2}{2\left (1 - \ttt\right )^2}$ is increasing in $\ttt$ and
\begin{equation*}
\forall~ 0 \leq \moww\leq 1,\quad \frac{\left(1 - \frac{2\moww}{3}\right)^2 + \left(\frac{2\moww}{3}\right)^2}{2\left(1 - \frac{2\moww}{3}\right)^2}\leq 1 + \frac{2\moww}{\left(2 - \moww\right)^2}\le \tau.
\end{equation*}

Therefore, $\mm\le 2$. If there are exactly two non-zero elements in $\vp$ (i.e., $\mm = 2$), then these elements must be equal.
\begin{claim}
\label{clm: m_2}
If $\mm=2$ and $p_1,p_2\ne 0$, then objective function is maximized when 
$\probone = \probtwo = \frac{\moww}{2}$.
\end{claim}
\begin{proof}
The objective function is 
\begin{equation*}
1 - \frac{1 + \frac{\probone}{\tau\left (1 - \probone\right )} + \frac{\probtwo}{\tau\left (1 - \probtwo\right )}}{\left (1 + \frac{\probone}{\tau\left (1 - \probone\right )}\right )\left (1 + \frac{\probtwo}{\tau\left (1 - \probtwo\right )}\right )} = \frac{1}{\tau\left (\left (\frac{1}{\probone} - 1\right )\left (\frac{1}{\probtwo} - 1\right )\tau + \frac{1}{\probone} + \frac{1}{\probtwo} - 2\right ) + 1}.
\end{equation*}
As both $\left (\frac{1}{\probone} - 1\right )\left (\frac{1}{\probtwo} - 1\right )$ and $\frac{1}{\probone} + \frac{1}{\probtwo}$ are minimized when $\probone = \probtwo = \frac{\momone}{2}$ the Claim follows. 
\end{proof}

Finally, when $\mm=1$, then our objective is simply equal to $0$. Therefore, the objective is maximized for $\mm=2$. By Claim~\ref{clm: m_2} it is when $\probone = \probtwo = \frac{\moww}{2}$, which gives exactly the desired upper bound of $\left(\frac{2 - \moww}{\moww}\cdot \tau + 1\right)^{-2}$ on $\Qtwoind(\tau)$. This concludes the proof of the main Claim~\ref{clm: range2}.
\end{proof}

%% file: sec/appendix3.tex
\begin{claim}
\label{clm: appD}
If $\max_i\{\Quani\left (\ww\right )\} \ge \pp$ and $\reserve\le\ww=1$,
then
$\frac{\tailrev}{\corerev} \leq \frac{2\pp^2 - 2\pp + 1}{\pp^3}\cdot \frac{e}{e - 1}.
$
\end{claim}
\begin{proof}
We will make use of the following simple calculation of integrals for any $\pp\in(0,1)$.
\begin{multline}
\label{eq: fact}
\int_0^1\frac{\pp \left (1 - \pp\cdot \right )}{\left (\left (1 - \pp \right )x + \pp \right )\left (\pp\cdot x + \left (1 - \pp \right )\right )}\d x = \frac{\pp(\pp - 1)\log(\frac{1}{\pp} - 1)}{2\pp - 1}\\
=\int_1^{+\infty}\frac{\pp \left (1 - \pp \right )}{\left (\left (1 - \pp \right )\cdot x + \pp \right )\left (\pp\cdot x + \left (1 - \pp \right )\right )}\d x .
\end{multline}
Hence, to prove our Claim, it suffices to show the following bounds.
\begin{align}
&\tailrev = \int_1^{+\infty}\Qtwoind\left (\tau\right )\d \tau\leq \frac{2\pp ^2 - 2\pp  + 1}{\pp ^3}\cdot \int_1^{+\infty}\frac{\pp \left (1 - \pp \right )}{\left (\left (1 - \pp \right )\cdot \tau + \pp \right )\left (\pp\cdot \tau + \left (1 - \pp \right )\right )}\d \tau,\label{eq: claim1}\\
&\corerev = \reserve \cdot \Qoneind\left (\reserve \right ) + \int_\reserve ^\ww\Qtwoind\left (\tau\right )\d \tau\geq \left (1 - \frac{1}{e}\right )\cdot \int_0^1\frac{\pp \left (1 - \pp \right )}{\left (\left (1 - \pp \right )\cdot \tau + \pp \right )\left (\pp\cdot \tau + \left (1 - \pp \right )\right )}\d \tau. 
\label{eq: claim2}
\end{align}

We first prove~\eqref{eq: claim1}. To this end, we derive upper bounds on $\Qtwoind(\tau)$ inside the integral for two separate ranges of $\tau$:
$\tau\geq 1 + \frac{1 - \pp }{\pp ^2}$ and $\tau\in(1,1 + \frac{1 - \pp }{\pp ^2})$. For $\tau\geq 1 + \frac{1 - \pp }{\pp ^2}$, we write
the same optimization program~\eqref{eq: optimization_C} as in  
Claim~\ref{clm: range2} in Appendix~\ref{app: upper_q2ind} and try the 
the same method of successive replacement of elements by
merging elements inside probability vector $\vp$. 
Now we have a useful extra condition from Claim~\ref{clm: appD} that $\max_i\{\Quani\left (\ww\right )\} \ge \pp$. Specifically, if  
$\vp$ has at least $\mm\ge 3$ non-zero entries, then the sum of two smaller non-zero elements $t=p_1+p_2\le \moww - \pp\le 1-\pp$. Hence, 
conditions of Claim~\ref{clm: m_gen} are satisfied when $\tau\ge 1 + \frac{1 - \pp }{\pp ^2}\ge 1+\frac{t}{(1-t)^2}$, which means that we can keep merging elements in $\vp$ until only two positive $\probi$ remain: $p_1\geq p$ and $p_2=\moww-p_1\le 1-p_1$. 
Then $\Qtwoind (\tau)=\Quani[1](\tau)\cdot\Quani[2](\tau)$,
where $\Quani[1](\tau)\le\frac{p_1}{(1-p_1)\cdot\tau+p_1}$ and $\Quani[2](\tau)\le\frac{1-p_1}{p_1\cdot\tau+1-p_1}$ by Lemma~\ref{lem: reg_quantile} for $\tau>1$. I.e.,
\begin{equation}
\label{eq: app_Tail_bound_2}
    \Qtwoind (\tau)\leq \frac{\pp_1 \left (1 - \pp_1 \right )}{\left (\left (1 - \pp_1 \right )\cdot \tau + \pp_1 \right )\left (\pp_1\cdot \tau + \left (1 - \pp_1 \right )\right )}
    \leq \frac{\pp \left (1 - \pp \right )}{\left (\left (1 - \pp \right )\cdot \tau + \pp \right )\left (\pp\cdot \tau + \left (1 - \pp \right )\right )},
\end{equation}
where the second inequality holds as it is monotone decreasing on $[\frac{1}{2}, 1]$.

For the interval $1< \tau< 1 + \frac{1 - \pp }{\pp ^2}$, 
similar to Lemma~\ref{lem: q2lb2}, we may assume that the first bidder $i=1$ has the largest quantile $\Quani(\tau)$. Let us also 
define
$
\widehat{Q}_1^{ind}\left (\tau\right )\eqdef \Prx[\vals\sim\distind]{|\{i>1: \vali\geq \tau\}|\geq 1}.
$
Then 
\begin{align*}
\Qtwoind\left (\tau\right )&\leq \widehat{Q}_1^{ind}\left (\tau\right ) \\
&\leq \sum_{i = 2}^n\Quani\left (\tau\right )\tag{Union bound}\\
&\leq \sum_{i = 2}^n\frac{\probi}{\left (1 - \probi\right )\cdot \tau + \probi}\tag{Lemma~\ref{lem: reg_quantile} for $\tau\ge 1$}\\
&\leq \frac{1 - \pp}{\pp\cdot\tau + \left (1 - \pp\right )}.\tag{Claim~\ref{pro: convex} for $\tau\ge 1$}
\end{align*}

Furthermore, when $\tau\leq 1 + \frac{1 - \pp }{\pp ^2}$ we have
\begin{multline}
\label{eq: app_Tail_bound_1}
 \Qtwoind(\tau )   \leq \frac{1 - \pp}{\pp\cdot\tau + \left (1 - \pp\right )} =\frac{\left (\left (1 - \pp \right )\cdot \tau + \pp \right )}{\pp}\cdot \frac{\pp (1 - \pp)}{\left((1 - \pp)\cdot \tau + \pp \right )(\pp\cdot \tau + (1 - \pp))}\\
 \leq \frac{2\pp^2 - 2\pp + 1}{\pp^3} \cdot \frac{\pp \left (1 - \pp \right )}{((1 - \pp )\cdot \tau + \pp ) (\pp\cdot \tau + \left (1 - \pp \right ))},
\end{multline}
where the last inequality holds as $\frac{\left (1 - \pp\right )\cdot \left (1 + (1 - \pp)/\pp^2\right ) + \pp}{\pp} = \frac{2\pp^2 - 2\pp + 1}{\pp^3}$ and $\tau\le 1 + \frac{1 - \pp }{\pp ^2}$.

Now we can conclude the proof of~\eqref{eq: claim1} by combining \eqref{eq: app_Tail_bound_1} and~\eqref{eq: app_Tail_bound_2} upper bounds on $\Qtwoind(\tau)$.
\begin{multline*}
\tailrev = 
\int_1^{1 + \frac{1 - \pp }{\pp ^2}}\Qtwoind\left (\tau\right )\d \tau
+\int_{1 + \frac{1 - \pp }{\pp ^2}}^{+\infty}\Qtwoind\left (\tau\right )\d \tau
\le 
\frac{2\pp ^2 - 2\pp  + 1}{\pp ^3}\cdot
\int_1^{1 + \frac{1 - \pp }{\pp ^2}}\frac{\pp \left (1 - \pp \right )}{\left (\left (1 - \pp \right )\cdot \tau + \pp \right )\left (\pp\cdot \tau + \left (1 - \pp \right )\right )}\d \tau+\\
1\cdot\int_{1 + \frac{1 - \pp }{\pp ^2}}^{+\infty}\frac{\pp \left (1 - \pp \right )}{\left (\left (1 - \pp \right )\cdot \tau + \pp \right )\left (\pp\cdot \tau + \left (1 - \pp \right )\right )}\d \tau
\le
\frac{2\pp ^2 - 2\pp  + 1}{\pp ^3}\cdot \int_1^{+\infty}\frac{\pp \left (1 - \pp \right )}{\left (\left (1 - \pp \right )\cdot \tau + \pp \right )\left (\pp\cdot \tau + \left (1 - \pp \right )\right )}\d \tau,
\end{multline*}
where the last inequality holds true, since $\frac{2\pp^2 - 2\pp + 1}{\pp^3}\geq 1$.
We finish the proof of Claim~\ref{clm: appD} by proving~\eqref{eq: claim2}. 
\begin{align*}
\Qtwoind(\tau)&\ge q_1(\tau)\cdot \widehat{Q}_1^{ind}\left (\tau\right )\\
&\geq q_1\left (\tau\right ) \cdot \left (1 - \frac{1}{e}\right )\cdot \sum_{i = 2}^n\Quani\left (\tau\right )\tag{Correlation gap}\\
&\geq\left (1 - \frac{1}{e}\right )\cdot \frac{\pp }{\left (1- \pp \right )\cdot \tau + \pp }\cdot \sum_{i = 2}^n\frac{\probi}{\left (1- \probi\right )\cdot \tau + \probi}\tag{Lemma~\ref{lem: reg_quantile} for $\tau\le 1$}\\
&\geq \left (1 - \frac{1}{e}\right )\cdot\frac{\pp \left (1 - \pp \right )}{\left (\left (1 - \pp \right )\cdot \tau + \pp \right )\left (\pp\cdot \tau + \left (1 - \pp \right )\right )}.\tag{Claim~\ref{pro: convex} for $\tau\le 1$}
\end{align*}
Taking the integral over $\tau\in(0,1)$ on both sides leads to ~\eqref{eq: claim2}. We complete the proof of Claim~\ref{clm: appD} 
by combining equations (\ref{eq: fact}-\ref{eq: claim2}).
\end{proof}